\documentclass[journal,romanappendices]{IEEEtran}
\usepackage{mathpple}
\usepackage{times}
\usepackage{amsmath,amsthm}  
\usepackage{amssymb}  
\usepackage[mathscr]{eucal}
\usepackage{cite}     
\usepackage{comment}  
\usepackage{upref}
\usepackage{amsfonts}
\usepackage{verbatim}
\usepackage[dvipsnames,usenames]{color}
\usepackage{enumerate}
\usepackage{graphicx}
\usepackage{latexsym}
\usepackage{bm}
\usepackage{multirow}
\usepackage{dsfont}
\usepackage{amsthm,xpatch}
\usepackage{mathtools}
\usepackage{bbm}
\usepackage{latexsym}
\usepackage{accents}
\usepackage{eufrak}
\usepackage{psfrag}
\usepackage{color}
\usepackage{times,color}
\usepackage[usenames,dvipsnames,svgnames,table]{xcolor}
\usepackage{pgf}
\usepackage{tikz}
\usetikzlibrary{arrows, automata, positioning, calc}
\usepackage{cancel} 
\usepackage{caption}
\usepackage{subcaption}
\usepackage{rotating}
\usepackage[T1]{fontenc}
\usepackage{anyfontsize}
\usepackage[hyphens]{url}
\usepackage{hyperref}

\newcommand{\myLambda}{\begin{sideways}%
     \begin{sideways}$\mathrm{V}$\end{sideways}\end{sideways}}

\newtheorem{lemma}{Lemma}
\newtheorem{theorem}{Theorem}
\newtheorem{remark}{Remark}
\newtheorem{proposition}{Proposition}

\usepackage{algorithm,algpseudocode}

\makeatletter
\newcommand{\removelatexerror}{\let\@latex@error\@gobble}
\makeatletter
\newcommand{\proofpart}[2]{%
	\par
	\addvspace{\medskipamount}%
	\noindent\emph{Part #1: #2}\par\nobreak
	\addvspace{\smallskipamount}%
	\@afterheading
}
\makeatother

\makeatletter
\newcommand*{\transpose}{%
  {\mathpalette\@transpose{}}%
}
\newcommand*{\@transpose}[2]{%
  \raisebox{\depth}{$\m@th#1\intercal$}%
}
\makeatother

\renewcommand{\mathsf}[1]{#1}

\theoremstyle{definition}

\IEEEoverridecommandlockouts

\begin{document}

\newcommand{\SB}[3]{
\sum_{#2 \in #1}\biggl|\overline{X}_{#2}\biggr| #3
\biggl|\bigcap_{#2 \notin #1}\overline{X}_{#2}\biggr|
}

\newcommand{\Mod}[1]{\ (\textup{mod}\ #1)}

\newcommand{\overbar}[1]{\mkern 0mu\overline{\mkern-0mu#1\mkern-8.5mu}\mkern 6mu}

\makeatletter
\newcommand*\nss[3]{%
  \begingroup
  \setbox0\hbox{$\m@th\scriptstyle\cramped{#2}$}%
  \setbox2\hbox{$\m@th\scriptstyle#3$}%
  \dimen@=\fontdimen8\textfont3
  \multiply\dimen@ by 4             % 4x the default rule thickness
  \advance \dimen@ by \ht0
  \advance \dimen@ by -\fontdimen17\textfont2
  \@tempdima=\fontdimen5\textfont2  % x-height
  \multiply\@tempdima by 4
  \divide  \@tempdima by 5          % 80% of the x-height
  % Modifications are only necessary if the top of the subscript is not that high:
  \ifdim\dimen@<\@tempdima
    \ht0=0pt                        % don't let the subscript interfere
    \@tempdima=\fontdimen5\textfont2
    \divide\@tempdima by 4          % 25% of the x-height
    \advance \dimen@ by -\@tempdima % if >0, add to depth of superscript!
    \ifdim\dimen@>0pt
      \@tempdima=\dp2
      \advance\@tempdima by \dimen@
      \dp2=\@tempdima
    \fi
  \fi
  #1_{\box0}^{\box2}%
  \endgroup
  }
\makeatother

\makeatletter
\renewenvironment{proof}[1][\proofname]{\par
  \pushQED{\qed}%
  \normalfont \topsep6\p@\@plus6\p@\relax
  \trivlist
  \item[\hskip\labelsep
        \itshape
%    #1\@addpunct{.}]\ignorespaces% DELETED
    #1\@addpunct{:}]\ignorespaces% ADDED
}{%
  \popQED\endtrivlist\@endpefalse
}
\makeatother

\makeatletter
\newsavebox\myboxA
\newsavebox\myboxB
\newlength\mylenA

\newcommand*\xoverline[2][0.75]{%
    \sbox{\myboxA}{$\m@th#2$}%
    \setbox\myboxB\null% Phantom box
    \ht\myboxB=\ht\myboxA%
    \dp\myboxB=\dp\myboxA%
    \wd\myboxB=#1\wd\myboxA% Scale phantom
    \sbox\myboxB{$\m@th\overline{\copy\myboxB}$}%  Overlined phantom
    \setlength\mylenA{\the\wd\myboxA}%   calc width diff
    \addtolength\mylenA{-\the\wd\myboxB}%
    \ifdim\wd\myboxB<\wd\myboxA%
       \rlap{\hskip 0.5\mylenA\usebox\myboxB}{\usebox\myboxA}%
    \else
        \hskip -0.5\mylenA\rlap{\usebox\myboxA}{\hskip 0.5\mylenA\usebox\myboxB}%
    \fi}
\makeatother

\xpatchcmd{\proof}{\hskip\labelsep}{\hskip3.75\labelsep}{}{}

\pagestyle{plain}

\title{\fontsize{21}{28}\selectfont Single-Server Private Linear Transformation:\\ The Joint Privacy Case}

\author{Anoosheh Heidarzadeh, \emph{Member, IEEE}, Nahid Esmati, \emph{Student Member, IEEE},\\ and Alex Sprintson, \emph{Senior Member, IEEE}\thanks{This work is to be presented in part at the 2021 IEEE International Symposium on Information Theory, Melbourne, Australia, July 2021.}\thanks{The authors are with the Department of Electrical and Computer Engineering, Texas A\&M University, College Station, TX 77843 USA (E-mail: \{anoosheh, nahid, spalex\}@tamu.edu).}}

%\thanks{This material is based upon work supported by the National Science Foundation (NSF) under Grants No.~1718658 and 1642983.}

% 

%\thanks{This work was supported by the National Science Foundation under Grant No.~CNS-0954153 and the AFOSR under Contract No.~FA9550-13-1-0008.}

\maketitle 

\thispagestyle{plain}

\begin{abstract} 
This paper introduces the problem of Private Linear Transformation (PLT) which generalizes the problems of private information retrieval and private linear computation. 
The PLT problem includes one or more remote server(s) storing (identical copies of) $K$ messages and a user who wants to compute $L$ independent linear combinations of a $D$-subset of messages. 
The objective of the user is to perform the computation by downloading minimum possible amount of information from the server(s), while protecting the identities of the $D$ messages required for the computation. 
In this work, we focus on the single-server setting of the PLT problem when the identities of the $D$ messages required for the computation must be protected jointly. 
We consider two different models, depending on whether the coefficient matrix of the required $L$ linear combinations generates a Maximum Distance Separable (MDS) code. 
We prove that the capacity for both models is given by $L/(K-D+L)$, where the capacity is defined as the supremum of all achievable download rates. 
Our converse proofs are based on linear-algebraic and information-theoretic arguments that establish connections between PLT schemes and linear codes. 
We also present an achievability scheme for each of the models being considered. 
\end{abstract}

\vspace{-0.125cm}
\begin{IEEEkeywords}
Private Information Retrieval, Private Function Computation, Information-Theoretic Privacy, Single Server, Linear Transformation, Maximum Distance Separable Codes. 
\end{IEEEkeywords}

% ---tailored to the single-server setting---

%The goal of the user is to perform this computation by downloading minimum amount of information from the server, while protecting the privacy of the identities of the $D$ messages required for the computation from the server. 

%In addition, we discuss several extensions of the PLT problem, including PLT in the presence of side information, and     

%\vspace{-0.01cm}
\section{introduction}
In this work, %we consider the problem of private computation of multiple linear combinations which we refer to as 
we introduce the problem of \emph{Private Linear Transformation (PLT)}. 
This problem includes one or more remote server(s) %(or multiple remote servers that can collude arbitrarily) 
storing (identical copies of) a dataset consisting of $K$ messages; and a user who is interested in computing $L$ independent linear combinations of a $D$-subset of messages. 
The objective of the user is to recover the $L$ required linear combinations by downloading minimum possible amount of information from the server(s), while the identities of the $D$ messages required for the computation are not revealed to
%privately so that the identities of the $D$ data items required for the computation are not revealed to %(to some degree) 
the server(s). %, while minimizing the total amount of information downloaded from the server(s).
%Single-server and multi-server settings of PLT 
%This problem is motivated by the application of linear transformation for dimensionality reduction in machine learning. 
% SJ2018No2
The PLT problem generalizes the problems of Private Information Retrieval (PIR)~\cite{CGKS1995,SJ2017,BU17,BU2018,TSC2019,CWJ2018,MMM2019,SS2021} and Private Linear Computation (PLC)~\cite{SJ2018,MM2018}, which have recently received a significant attention from the research community.
% TM2019No2
In particular, PLT reduces to the PIR problem or the PLC problem when $L=D$ or $L=1$, respectively. 
This is because in PIR, the problem is to privately retrieve \emph{a $D$-subset of messages}, which is equivalent to privately computing \emph{$D$ independent linear combinations} of the $D$ desired messages; and in PLC, the problem is to privately compute \emph{one linear combination} of a $D$-subset of messages. 

The PLT problem is motivated by the need to protect the data access patterns in several Machine Learning (ML) applications such as linear transformation for dimensionality reduction, see, e.g.,~\cite{CG2015}, and training different linear regression or classification models in parallel, see, e.g.,~\cite{A2019,JY2020}. % and references therein. 
For instance, consider a dataset with $N$ data samples each with $K$ attributes, represented by a $K\times N$ data matrix. 
Suppose there is a user who wishes to implement an ML algorithm on a subset of $D$ selected attributes, while protecting the privacy of the selected attributes.  
When $D$ is large, the $D$-dimensional feature space is typically mapped onto a new subspace of lower dimension, say, $L$, and the ML algorithm operates on the new $L$-dimensional subspace instead. 
A commonly-used technique for dimensionality reduction is \emph{linear transformation}, where an $L\times D$ matrix is multiplied by the $D\times N$ submatrix of the $K\times N$ data matrix restricted to the $D$ selected attributes. 
This scenario matches the setup of the PLT problem, in which each message represents the $N$ data samples for one attribute, the labels of the selected attributes correspond to the identities of the messages required for the computation, and the transformation matrix is formed by the coefficient matrix of the required linear combinations.

In many practical scenarios, the dataset is stored on a single server, or multiple servers that belong to the same provider and can collude arbitrarily. 
Motivated by such scenarios, in this work we focus on the single-server setting of the PLT problem. 
A simple approach for PLT is to privately retrieve the messages required for the computation using a single-server PIR scheme, and then compute the required linear combinations locally. 
As shown in~\cite{KGHERS2020,HKS2019Journal,HKRS2019,KKHS32019,HKGRS2018,LG2018,HKS2018,HKS2019}, leveraging a prior side information about the dataset, in the single-server setting, the user can retrieve a single or multiple messages privately with a much lower download cost than the trivial scheme of downloading the entire dataset. 
(The advantages of side information in multi-server PIR were also studied in~\cite{T2017,WBU2018,WBU2018No2,CWJ2020,SSM2018,KKHS22019,KKHS12019,KH2021}.) %when the user initially has access to a subset of data items or a linear combination of them as a side information and the identities of these items are initially unknown to the server 
However, when there is no side information, a PIR-based approach is extremely expensive as the entire dataset must be downloaded in order to achieve information-theoretic privacy~\cite{CGKS1995}. 
%the PIR-based approach is, however, extremely expensive since, in the absence of any side information, all items in the dataset must be downloaded in order to achieve privacy~\cite{}. 
%Recently, similar results were also shown for single-server PLC~\cite{}. 
Another approach for PLT is to privately compute the required linear combinations separately via applying a single-server PLC scheme multiple times. 
(The multi-server PLC problem and its extensions were studied in~\cite{SJ2018,MM2018,OK2018,OLRK2018,TM2019No2,OLRK2020}.)
In~\cite{HS2019PC,HS2020}, it was shown that PLC can be performed more efficiently than PIR in terms of the download cost, regardless of whether the user has any side information or not. %This suggests that a PLC-based PLT scheme can outperform a PIR-based PLT scheme; 
However, a PLC-based approach may still lead to an unnecessary overhead due to the redundancy in the information being downloaded. 
This implies the need for novel PLT schemes with optimal download rate.

Different types of privacy can be considered for PLT. 
In this work, we focus on the PLT problem under a strong notion of privacy, called \emph{joint privacy}, which was also considered previously for PIR and PLC (see, e.g.,~\cite{BU2018,HKGRS2018,HS2020,ZTSP2021}).
We refer to this problem as \emph{PLT with Joint Privacy}, or \emph{JPLT} for short. 
The joint privacy requirement implies that 
%consider the strong notion of \emph{joint privacy}, %two notions of information-theoretic %information-theoretic 
%privacy: (i) \emph{joint privacy}---
%The joint privacy requirement implies that 
%any subset of data items of size $D$ must be equally likely to be the support set of the required linear combinations.
the identities of all $D$ messages required for the computation must be kept private jointly. %, and no information about the correlation between them must be leaked. %~\cite{BU2018,HKGRS2018,HS2020}. 
%We refer to this problem as \emph{PLT with Joint Privacy}, or \emph{JPLT} for short. In a parallel submission~\cite{EHS2021Individual}, we have considered a weaker privacy condition, called \emph{individual privacy}, where the identity of every individual item required for the computation must be kept private, but some information about the correlation between these items may be leaked. %Note that joint privacy implies individual privacy, but not vice versa. 
%\footnote{}   
%the identities of the data items in the support set of the demanded linear combinations, %combinations~\cite{HKGRS2018,LG2018,BU17,BU2018,HS2020}; %We refer to this problem as \emph{Jointly-Private Linear Transformation (JPLT)}. 
%We measure the efficiency of a JPLT scheme by its \emph{download rate}, and 
This type of privacy is of practical importance in the scenarios in which the correlation between the identities of messages required for the computation need to be kept private. 
For instance, the user may want to compute a linear combination of two vectors, and the server must not learn which pair of vectors were required for the computation. 

In a parallel work~\cite{HES2021IndividualJournal}, we have considered a relaxed version of joint privacy, called \emph{individual privacy}, which was recently introduced for PIR and PLC (see, e.g.,~\cite{HKRS2019,HS2020}). 
The individual privacy condition ensures that the identity of every individual message required for the computation is kept private. 
In contrast to joint privacy, individual privacy finds application in the scenarios in which the correlation between the identities of the required messages does not need to be protected. 
For example, the dataset may contain information about individuals, and the user is required to hide information from the server on whether the data belonging to an individual was used in the computation. 
%In a parallel work~\cite{}, we have considered a weaker notion of privacy, called \emph{individual privacy}, where every message index must be equally likely to belong to the demand's support index set. 
   
Unlike the privacy requirements for the multi-server PLC problem in~\cite{SJ2018,MM2018} and the multi-server Private Monomial Computation problem in~\cite{YLR2020}, % (which is asymptotically equivalent to the multi-server PLC problem as the field size grows unbounded), 
joint and individual privacy are to protect the data access patterns, and not the values of the coefficients (or the exponents) in the required linear combination (or the required monomial function). 
These types of access privacy are inspired by several real-world scenarios. 
For example, protecting the identities of the selected attributes %(jointly or individually) 
in the application of linear transformation for dimensionality reduction may prevent the server from learning the user's data access patterns which, in turn, can be instrumental for hiding user's algorithms, preferences, and objectives from the server. %, to a satisfactory degree.

\subsection{Main Contributions}
We consider two different models, referred to as \emph{Model~I} and \emph{Model~II}, for the JPLT problem.
In Model~I, it is assumed that the coefficient matrix of the required linear combinations is maximum distance separable (MDS), whereas in Model~II, it is assumed that the coefficient matrix has full rank (but it may or may not be MDS).\footnote{A $k\times n$ matrix is said to be MDS iff it generates an $[n,k]$ MDS code.} % (but it may or may not generate an MDS code).
Model~I is motivated by the scenarios in which the combination coefficients are chosen purposely to form an MDS matrix, %, e.g., in the design of codes for distributed storage~\cite{RVBKSK2020}, 
or the coefficient matrix is randomly generated over the field of real numbers or a finite field of large size,\footnote{A direct application of Schwartz-Zippel lemma~\cite{S1980,Z1979} shows that a matrix whose entries are randomly chosen from a sufficiently large field is MDS with high probability.} e.g., when applying random linear transformation for dimensionality reduction~\cite{BM2001}.
Model~II, on the other hand, finds application in the scenarios in which the size of the operating field is relatively small, e.g., due to the computational complexity considerations, and the number of rows ($L$) and the number of columns ($D$) of the coefficient matrix are such that $\binom{D}{L}$ is large, e.g., when a large reduction factor is required in dimensionality reduction.
We refer to the JPLT problem under Model~I or Model~II as the \emph{JPLT-I} or \emph{JPLT-II} problem, respectively.

In this work, we characterize the capacity of the JPLT-I and JPLT-II problems, where the capacity of JPLT-I (or JPLT-II) is defined as the supremum of download rates over all JPLT-I (or JPLT-II) schemes.
In particular, we prove that the capacity of both problems is given by $L/(K-D+L)$.
This result is particularly interesting because it shows that JPLT can be performed more efficiently than applying a PIR-based or a PLC-based approach for privately computing multiple linear combinations simultaneously.  
For each problem, we prove the converse by using a mix of linear-algebraic and information-theoretic arguments. %, relying primarily on a necessary condition for JPLT-I or JPLT-II schemes. 
Our technique for proving the converse for the JPLT-II problem is more general and is applicable to the JPLT-I problem. 
However, this technique is based on proof-by-contradiction. 
On the other hand, our proof technique for the JPLT-I problem is a constructive proof which also gives insight into the design of an achievability scheme. 
For the JPLT-I problem, we propose an achievability scheme, termed the \emph{Specialized MDS Code protocol}, which is based on the idea of extending an MDS code.\footnote{Extending a code is performed by adding new columns to the generator matrix of the code.} 
For the JPLT-II problem, we propose a different achievability scheme, termed the \emph{Specialized Augmented Code protocol}. 
This scheme is based on augmenting a non-MDS code by an MDS code.\footnote{Augmenting a code is performed by adding new rows to the generator matrix of the code.}

\subsection{Notation}
We denote random variables and their realizations by bold-face and regular symbols, respectively. 
We denote sets, vectors, and matrices by roman font, %(\texttt{\textbackslash mathrm}), 
and denote collections of sets, vectors, or matrices by blackboard bold roman font. %(\texttt{\textbackslash mathbbm}). %, and fields are denoted by blackboard bold sans serif font (\texttt{\textbackslash mathbbmss}). 
%For any matrix $\mathrm{M}$, we denote by $\mathrm{rank}(\mathrm{M})$ the rank of the matrix $\mathrm{M}$. 
%For any events $E_1,E_2$, we denote by $\Pr(E_1)$ and $\Pr(E_1|E_2)$ the probability of $E_1$ and the conditional probability of $E_1$ given $E_2$, respectively. 
For any random variables $\mathbf{X},\mathbf{Y}$, we denote by $H(\mathbf{X})$ and $H(\mathbf{X}|\mathbf{Y})$ the entropy of $\mathbf{X}$ and the conditional entropy of $\mathbf{X}$ given $\mathbf{Y}$, respectively. 
For any integer $n\geq 1$, we denote $\{1,\dots,n\}$ by $[n]$, and for any integers $1<n<m$, we denote $\{n,n+1,\dots,m\}$ by $[n:m]$.
We denote the binomial coefficient $\binom{n}{k}$ by $C_{n,k}$.  
%Throughout this paper, random variables and their realizations are denoted by bold-face symbols (e.g., $\mathbf{X},\mathbf{W}$) and non-bold-face symbols (e.g., $\mathrm{X},\mathrm{W}$), respectively. 

\section{Problem Setup}\label{sec:SN}
\subsection{Models and Assumptions}
Let $q$ be an arbitrary prime power, and let $N\geq 1$ be an arbitrary integer. 
Let $\mathbbmss{F}_q$ be a finite field of order $q$, and let $\mathbbmss{F}_{q}^{N}$ be the $N$-dimensional vector space over $\mathbbmss{F}_q$.
Let $B\triangleq N\log_2 q$. 
Let $K,D,L\geq 1$ be integers such that ${L\leq D\leq K}$. 
%Let ${[K]\triangleq \{1,\dots,K\}}$ be the set of integers from $1$ to $K$. 
We denote by $\mathbbm{W}$ the set of all $D$-subsets (i.e., all subsets of size $D$) of $[K]$. %denote by $\mathscr{V}$ the collection of all $L\times D$ matrices (with entries from $\mathbbmss{F}_p$) with $L$ linearly independent rows and $D$ nonzero columns. 
Also, we denote by $\mathbbm{V}_{I}$ the set of all $L\times D$ matrices $\mathrm{V}$ with entries in $\mathbbmss{F}_q$ that are MDS, i.e., every $L\times L$ submatrix of $\mathrm{V}$ is invertible, and denote by $\mathbbm{V}_{I\hspace{-0.04cm}I}$ the set of all $L\times D$ matrices $\mathrm{V}$ with entries in $\mathbbmss{F}_q$ that have full rank, i.e., $\mathrm{rank}(\mathrm{V}) = L$. 

Consider a server that stores $K$ messages ${X_1,\dots,X_K}$, where $X_i\in \mathbbmss{F}_q^{N}$ for $i\in [K]$ is a row-vector of length $N$. 
Let ${\mathrm{X}\triangleq [X_1^{\transpose},\dots,X_K^{\transpose}]^{\transpose}}$. Note that $\mathrm{X}$ is a matrix of size $K\times N$.
For every ${\mathrm{S}\subset [K]}$, we denote by $\mathrm{X}_{\mathrm{S}}$ the submatrix of $\mathrm{X}$ restricted to its rows indexed by $\mathrm{S}$, i.e., $\mathrm{X}_{\mathrm{S}} = [X_{i_1}^{\transpose},\dots,X_{i_{s}}^{\transpose}]^{\transpose}$, where ${\mathrm{S} = \{i_1,\dots,i_{s}\}}$. 
Note that $\mathrm{X}_{\mathrm{S}}$ is a matrix of size $|\mathrm{S}|\times N$, where $|\mathrm{S}|$ denotes the size of $\mathrm{S}$.
Consider a user who wishes to compute $L$ linear combinations of $D$ messages, namely, $\mathrm{v}_1 \mathrm{X}_{\mathrm{W}},\dots,\mathrm{v}_L \mathrm{X}_{\mathrm{W}}$, where $\mathrm{W}\in \mathbbm{W}$ is the index set of the $D$ messages required for the computation, and $\mathrm{v}_l$ for each $l\in [L]$ is a row-vector of length $D$ with entries in $\mathbbmss{F}_q$, denoting the coefficient vector of the $l$th required linear combination. 
We represent the collection of the required linear combinations in the matrix form as $\mathrm{Z}^{[\mathrm{W},\mathrm{V}]}\triangleq \mathrm{V}\mathrm{X}_{\mathrm{W}}=\mathrm{U}\mathrm{X}$, where ${\mathrm{V}= [\mathrm{v}_{1}^{\transpose},\dots,\mathrm{v}_{L}^{\transpose}]^{\transpose}}$ is an $L\times D$ matrix with entries in $\mathbbmss{F}_q$, denoting the coefficient matrix pertaining to the required linear combinations, and $\mathrm{U}$ is an $L\times K$ matrix such that the submatrix of $\mathrm{U}$ restricted to the columns indexed by $\mathrm{W}$ is equal to $\mathrm{V}$, and the rest of the columns of $\mathrm{U}$ are all-zero. 
Note that $\mathrm{Z}^{[\mathrm{W},\mathrm{V}]}$ is a matrix of size $L\times N$ with entries in $\mathbbmss{F}_q$. 
We refer to $\mathrm{Z}^{[\mathrm{W},\mathrm{V}]}$ as the \emph{demand}, $\mathrm{W}$ as the \emph{support of the demand}, $\mathrm{V}$ as the \emph{coefficient matrix of the demand}, $\mathrm{U}$ as the \emph{global coefficient matrix of the demand}, $D$ as the \emph{support size of the demand}, and $L$ as the \emph{dimension of the demand}.  

In this work, we consider two different models:
\begin{itemize}
\item \emph{Model~I:} $\mathrm{v}_l \mathrm{X}_{\mathrm{W}}$'s are $L$ \emph{MDS-coded} linear combinations of the $D$ messages indexed by $\mathrm{W}$, i.e., $\mathrm{V}\in \mathbbm{V}_{I}$. 
%This model is motivated by the scenarios in which either the combination coefficients are chosen purposely to form an MDS matrix, e.g., in the design of codes for distributed storage, or the coefficient matrix is randomly generated over the field of real numbers or a finite field of large size, e.g., when applying random linear transformation for dimensionality reduction. 
%In particular, using the Schwartz–Zippel lemma, it can be shown that the larger is the field size $q$, the higher is the probability that a randomly generated transformation matrix with entries in $\mathbbmss{F}_q$ is MDS.
\item \emph{Model~II:} $\mathrm{v}_l \mathrm{X}_{\mathrm{W}}$'s are $L$ \emph{linearly independent} (but not necessarily MDS-coded) linear combinations of the $D$ messages indexed by $\mathrm{W}$, i.e., $\mathrm{V}\in \mathbbm{V}_{I\hspace{-0.04cm}I}$. 
%This model finds application in the scenarios in which the size of the operating field is relatively small, e.g., due to the computational complexity considerations, and the number of rows ($L$) and the number of columns ($D$) of the coefficient matrix are such that $\binom{D}{L}$ is relatively large, e.g., when a very large reduction factor is required in dimensionality reduction. 
\end{itemize}

Throughout, we make the following assumptions: 
\begin{enumerate}
\item $\mathbf{X}_1,\dots,\mathbf{X}_K$ are independently and uniformly distributed over $\mathbbmss{F}_{q}^{N}$. 
%Thus, ${H(\mathbf{X}_{i})=B}$ for $i\in [K]$, where $B\triangleq \log_2 q$, and more generally, 
Thus, $H(\mathbf{X})=KB$, and ${H(\mathbf{X}_{\mathrm{S}})= |\mathrm{S}| B}$ for every ${\mathrm{S}\subset [K]}$. 
Moreover, $H(\mathbf{Z}^{[\mathrm{W},\mathrm{V}]})=LB$ for Model~I and Model~II.
\item $\mathbf{W}, \mathbf{V}, \mathbf{X}$ are independent random variables. 
\item $\mathbf{W}$ is distributed uniformly over ${\mathbbm{W}}$. 
\item $\mathbf{V}$ is distributed uniformly over ${\mathbbm{V}_{I}}$ or ${\mathbbm{V}_{I\hspace{-0.04cm}I}}$ for Model~I or Model~II, respectively. % that are \emph{Maximum Distance Separable (MDS)}, i.e., every $L\times L$ submatrix of $\mathrm{V}$ is invertible; 
\item The demand's support size $D$ and dimension $L$, %the underlying model (i.e., whether $\mathbf{V}$ is distributed uniformly over $\mathbbm{V}_{I}$ or $\mathbbm{V}_{I\hspace{-0.04cm}I}$), 
and the distributions of $\mathbf{W}$ and $\mathbf{V}$ are initially known by the server, whereas the realizations $\mathrm{W}$ and $\mathrm{V}$ are initially unknown to the server.
\end{enumerate}

\subsection{Privacy and Recoverability Conditions}
Given $\mathrm{W}$ and $\mathrm{V}$, the user generates a query $\mathrm{Q}^{[\mathrm{W},\mathrm{V}]}$, simply denoted by $\mathrm{Q}$, and sends it to the server. 
For simplicity, we denote $\mathbf{Q}^{[\mathbf{W},\mathbf{V}]}$ by $\mathbf{Q}$. 
The query $\mathrm{Q}$ is a deterministic or stochastic function of $\mathrm{W}, \mathrm{V}$. 
In the case of a deterministic query, $H(\mathbf{Q}|\mathbf{W},\mathbf{V})=0$, and in the case of a stochastic query, $H(\mathbf{Q}|\mathbf{W},\mathbf{V},\mathbf{R})=0$, where $\mathrm{R}$ is a random key generated by the user (independently from $\mathrm{W},\mathrm{V},\mathrm{X}$), and unknown to the server. 

%initially generated by the user---independently from $\mathrm{W},\mathrm{V},\mathrm{X}$, and initially unknown to the server.
%Here, $\mathbf{Q}$ denotes $\mathbf{Q}^{[\mathbf{W},\mathbf{V}]}$. 

%where $\mathbf{Q}^{[\mathbf{W},\mathbf{V}]}$ is denoted by $\mathbf{Q}$.

%, and potentially a random key $\mathrm{R}$ that is generated by the user in advance---independently from $\mathrm{W},\mathrm{V},\mathrm{X}$, and is initially unknown to the server. 
%That is, $H(\mathbf{Q}|\mathbf{W},\mathbf{V},\mathbf{R})=0$, where $\mathbf{Q}^{[\mathbf{W},\mathbf{V}]}$ is denoted by $\mathbf{Q}$. 

Given the query $\mathrm{Q}$, every $D$-subset of message indices must be equally likely to be the demand's support $\mathbf{W}$, i.e., for every $\tilde{\mathrm{W}} \in \mathbbm{W}$, it must hold that
\begin{equation*}
\Pr (\mathbf{W}=\tilde{\mathrm{W}}|\mathbf{Q}=\mathrm{Q})=\Pr(\mathbf{W}=\tilde{\mathrm{W}})={1}/{C_{K,D}}. 
\end{equation*} We refer to this condition as the \emph{joint privacy condition}.

Upon receiving the query $\mathrm{Q}$, the server generates an answer $\mathrm{A}^{[\mathrm{W},\mathrm{V}]}$, simply denoted by $\mathrm{A}$, and sends it back to the user. 
For simplicity, we denote $\mathbf{A}^{[\mathbf{W},\mathbf{V}]}$ by $\mathbf{A}$.
The answer $\mathrm{A}$ is a deterministic function of $\mathrm{Q}$ and $\mathrm{X}$.
That is, $H(\mathbf{A}|\mathbf{Q},\mathbf{X})=0$. %, where $\mathbf{A}^{[\mathbf{W},\mathbf{V}]}$ is denoted by $\mathbf{A}$. 

The answer $\mathrm{A}$, the query $\mathrm{Q}$, and the realizations $\mathrm{W}, \mathrm{V}$ must collectively enable the user to retrieve the demand
$\mathrm{Z}^{[\mathrm{W},\mathrm{V}]}$, i.e., 
\[H(\mathbf{Z}| \mathbf{A},\mathbf{Q}, \mathbf{W},\mathbf{V})=0,\] where $\mathbf{Z}^{[\mathbf{W},\mathbf{V}]}$ is denoted by $\mathbf{Z}$. 
We refer to this condition as the \emph{recoverability condition}. %Note that $H(\mathbf{Z})=LB$.

\subsection{Problem Statement}
The problem is to design a protocol for generating a query $\mathrm{Q}^{[\mathrm{W},\mathrm{V}]}$ and the corresponding answer $\mathrm{A}^{[\mathrm{W},\mathrm{V}]}$ for any given $\mathrm{W}$ and $\mathrm{V}$
such that the joint privacy and recoverability conditions are satisfied.
We refer to this problem as single-server \emph{Private Linear Transformation (PLT) with Joint Privacy}, or \emph{JPLT} for short. 
The JPLT problem under Model~I (or Model~II) is referred to as the \emph{JPLT-I} (or \emph{JPLT-II}) problem, and a protocol for JPLT-I (or JPLT-II) is referred to as a \emph{JPLT-I} (or \emph{JPLT II}) \emph{protocol}. 
A protocol is called \emph{linear} if the server's answer to the user's query consists only of linear combinations of the messages; otherwise, the protocol is called \emph{non-linear}.  

%Following the convention in the PIR and PLC literature, 
We measure the efficiency of a JPLT-I or JPLT-II protocol by its \emph{rate}---defined as the ratio of the entropy of the demand (i.e., $H(\mathbf{Z})=LB$) to the entropy of the answer (i.e., $H(\mathbf{A})$). 
We define the \emph{capacity} of JPLT-I or JPLT-II as the supremum of rates over all JPLT-I or JPLT-II protocols, respectively.
In this work, our goal is to characterize %derive (tight) lower and upper bounds on 
the capacity of these settings in terms of $K,D,L$. 
Note that the capacity may also depend on the field size $q$ in general. 
Notwithstanding, in this work we are interested in characterizing the supremum of rates over all protocols and all $q$.\footnote{Our converse bounds hold for any $q$, and our achievability schemes achieve these converse bounds when $q$ is sufficiently large, depending on $K,D,L$.}

\begin{figure*}[t!]
\centering
\begin{subfigure}[t]{.5\textwidth}
  \centering
  \includegraphics[width=.75\linewidth]{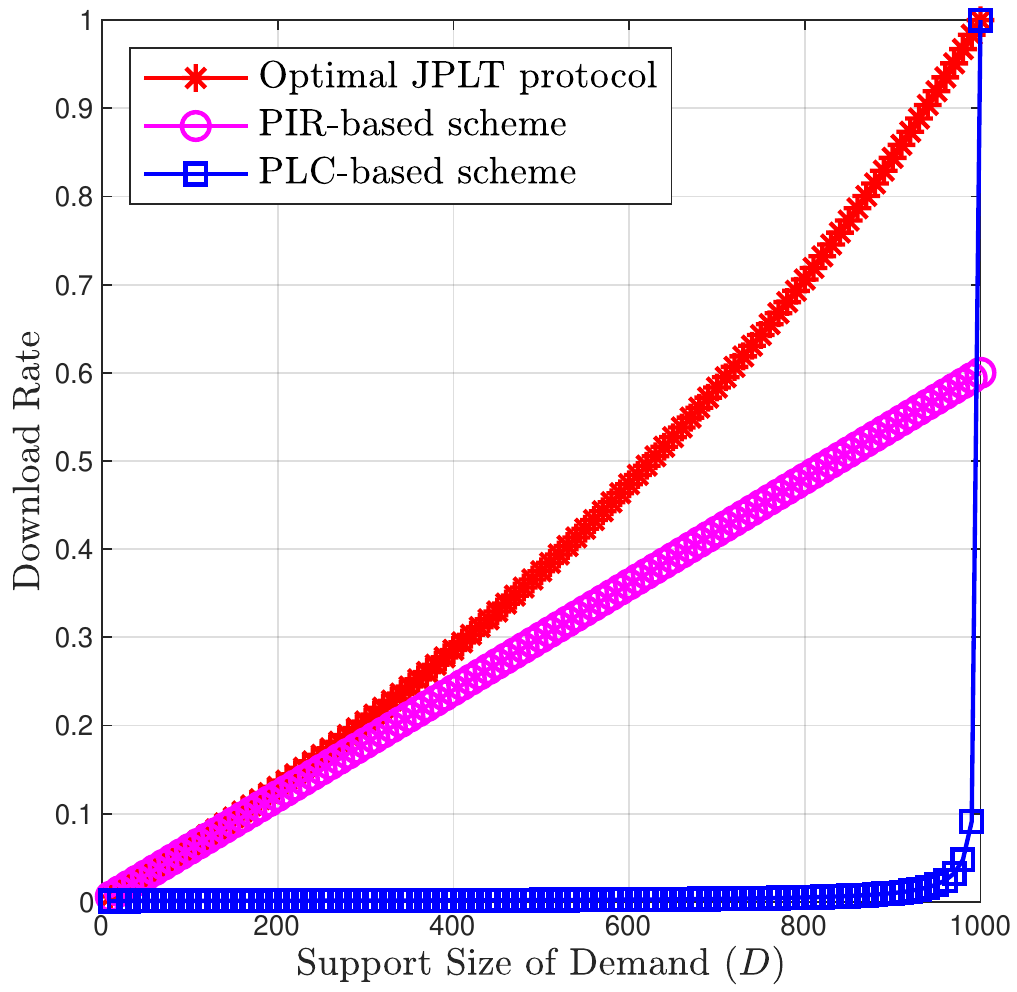}
  \caption{$K=1000$, $L/D=0.6$}
  \label{fig:JPLT1}
\end{subfigure}%
\begin{subfigure}[t]{.5\textwidth}
  \centering
  \includegraphics[width=.75\linewidth]{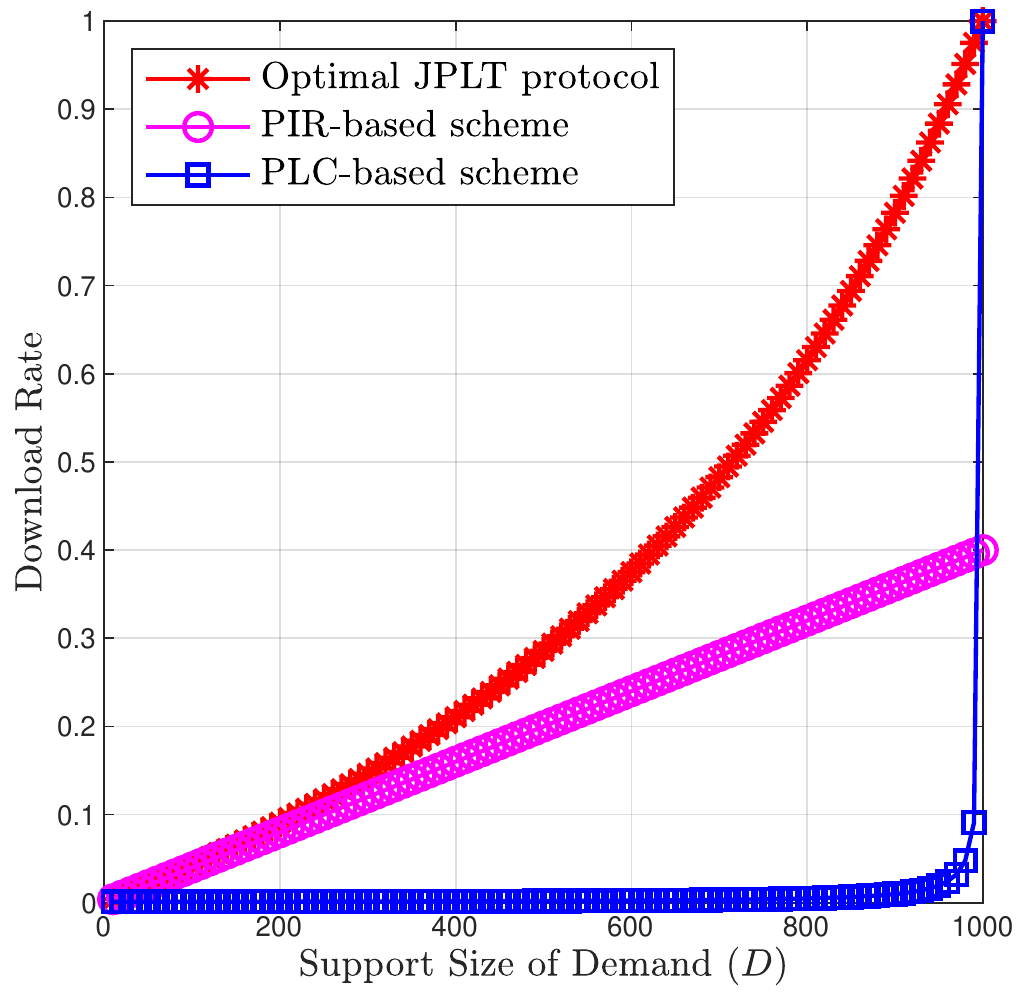}
  \caption{$K=1000$, $L/D=0.4$}
  \label{fig:JPLT2}
\end{subfigure}%\vspace{0.2cm}
\caption{The download rate of an optimal JPLT protocol and the PIR-based and PLC-based schemes.}
\label{fig:JPLT}
\end{figure*}

\section{Main Results}
In this section, we present our main results. % for the JPLT-I and JPLT-II settings. 
Theorems~\ref{thm:JPLT1} and~\ref{thm:JPLT2} characterize the capacity of JPLT-I and JPLT-II, respectively. 
The proofs are given in Sections~\ref{sec:JPLT1} and~\ref{sec:JPLT2}, respectively. 
%characterizes the capacity of the JPLT-II setting. 

\begin{theorem}\label{thm:JPLT1}
For the JPLT-I setting with $K$ messages, demand's support size $D$, and demand's dimension $L$, the capacity is given by $L/(K-D+L)$. 	
\end{theorem}

The proof of converse is based on a mix of linear-algebraic and information-theoretic arguments. %relying on the joint privacy and recoverability conditions. 
A key ingredient of the proof is the result of Lemma~\ref{lem:NCJPLT1} which follows from the joint privacy and recoverability conditions for Model~I. %mostly on a condition for JPLT-I protocols---provided by Lemma~\ref{lem:NCJPLT1}. %(see Section~\ref{subsec:JPLT-Conv}). 
%(see Section~\ref{sec:JPLT-Conv}). 
The converse bound naturally serves as an upper bound on the rate of any JPLT-I protocol. 
We prove the achievability by designing a linear JPLT-I protocol, termed the \emph{Specialized MDS Code protocol}, that achieves the converse bound. %(see Section~\ref{subsec:JPLT-Ach}). 
%(see Section~\ref{sec:JPLT-Ach}). 
This protocol generalizes those in~\cite{HKGRS2018} and~\cite{HS2019PC} for single-server PIR and PLC with joint privacy (when the user has no prior side information about the content of the messages available at the server), and is based on the idea of extending the MDS code generated by the coefficient matrix of the demand. 
%In particular, for the cases in which the coefficient matrix of the demand generates a Generalized Reed-Solomon (GRS) code, we give an explicit construction of a GRS code that contains a specific collection of codewords---specified by the demand's support and coefficient matrix. 
%We refer to this protocol as the \emph{Specialized GRS Code protocol}.

\begin{comment}
\begin{figure*}[t!]
\centering
\includegraphics[width=\textwidth]{figures/JPLT_vs_PIR_PLC.pdf}
\caption{The download rate of the proposed scheme and the PIR-based and PLC-based schemes.}
\label{fig:JPLT}
\end{figure*}
\end{comment}

\begin{theorem}\label{thm:JPLT2}
For the JPLT-II setting with $K$ messages, demand's support size $D$, and demand's dimension $L$, the capacity is given by $L/(K-D+L)$. 	
\end{theorem}

We prove the converse for the JPLT-II problem %by using similar arguments as in our converse proof for the JPLT-I problem, 
by relying on 
%except that in this case we rely on 
the result of Lemma~\ref{lem:NCJPLT2} which follows from the joint privacy and recoverability conditions for Model~II. %a necessary condition for JPLT-II protocols---stated in Lemma~\ref{lem:NCJPLT2}. 
The proof is by the way of contradiction, and is also applicable to the JPLT-I problem.  
That said, for the JPLT-I problem we present a different converse proof based on construction, which also gives insight into the design of an achievability scheme.  
Note that our constructive proof technique does not extend to the JPLT-II problem. 
This is because the construction we propose in the proof relies on the fact that MDS matrices do not contain any all-zero columns. 
This condition, however, does not always hold for full (row-) rank matrices. %, this is while an $L\times D$ full-rank matrix can contain up to $D-L$ all-zero columns. 

%It should be noted that our proof technique for the case of MDS matrices does not extend to the case of full-rank matrices. 
%This is because in the former case there cannot be any all-zero column, whereas in the latter case, up to $D-L$ columns can be all-zero. 

%The converse proof for the case of full-rank matrices is applicable to the case of MDS matrices. 

%That said, we present a constructive proof for the JPLT-I problem that gives insight into the design of an achievability scheme. 

%This is in contrast to our proof by contradiction for the JPLT-II problem. 

To prove the achievability result, we propose a linear JPLT-II protocol, termed the \emph{Specialized Augmented Code protocol}, that achieves the converse bound. 
This protocol is based on the idea of augmenting the global coefficient matrix of the demand by an MDS code. 
The main difference between our achievability schemes for JPLT-I and JPLT-II is that unlike the Specialized MDS Code protocol, the Specialized Augmented Code protocol does not necessarily generate an MDS code.

\begin{remark}\label{rem:JPLT1}
\emph{In~\cite{HS2019PC}, it was shown that the rate ${1/(K-D+1)}$ is achievable for single-server PLC with joint privacy when the user has no prior side information about the messages available at the server. 
The optimality of this rate, however, was not shown. 
The results of Theorems~\ref{thm:JPLT1} and~\ref{thm:JPLT2} for ${L=1}$ prove the optimality of this rate.   
For $L=D$, the JPLT-I and JPLT-II problems are equivalent to the problem of single-server PIR without any prior side information when joint privacy is required. As was shown in~\cite{HKGRS2018}, an optimal solution for this problem is to download the entire dataset. 
This is consistent with the results of Theorems~\ref{thm:JPLT1} and~\ref{thm:JPLT2} for ${L=D}$.}	
\end{remark}

\begin{remark}\label{rem:JPLT2}
\emph{The results of Theorems~\ref{thm:JPLT1} and~\ref{thm:JPLT2} show that %, when there is only a single server and there is no prior side information available at the user, 
JPLT-I and JPLT-II, collectively referred to as JPLT, can be performed more efficiently than using either of the following PIR-based and PLC-based approaches: 
(i) retrieving the messages required for the user's computation using a single-server multi-message PIR scheme that achieves joint privacy~\cite{HKGRS2018}, and then computing the required linear combinations locally, or 
(ii) computing each of the required linear combinations separately via applying a single-server PLC scheme that achieves joint privacy~\cite{HS2019PC}. 
Note that the optimal rate for the PIR-based or PLC-based scheme is $L/K$ or $1/(K-D+1)$, respectively, whereas an optimal JPLT protocol achieves the rate ${L/(K-D+L)}$. 
Fig.~\ref{fig:JPLT} depicts the download rate of an optimal JPLT protocol, the PIR-based scheme, and the PLC-based scheme, for different values of $D\in \{10,20,\dots,1000\}$, where $K=1000$, and $L/D=0.6$ (left plot) or $L/D=0.4$ (right plot). %As can be seen, the proposed scheme outperforms the other two schemes. 
As can be seen in Fig.~\ref{fig:JPLT}, for a fixed ratio $L/D$, %---which corresponds to a fixed dimensionality reduction factor , 
the advantage of an optimal JPLT protocol over the PIR-based scheme is more pronounced as $D$ increases. 
For instance, for $L/D=0.4$, the rate of an optimal JPLT protocol is about $15\%$ and $30\%$ more than that of the PIR-based scheme for $D=250$ and $D=500$, respectively.
%Comparing an optimal JPLT protocol and the PLC-based scheme, 
It can also be seen in Fig.~\ref{fig:JPLT} that when the ratio $L/D$ is fixed, the gap between the rate of an optimal JPLT protocol and the rate of the PLC-based scheme increases as $D$ increases up to a threshold very close to $K$; 
and beyond this threshold, the gap decreases rapidly as $D$ increases up to $K$. 
In addition, a comparison of the left and right plots in Fig.~\ref{fig:JPLT} shows that for a fixed value of $D$, the smaller is the ratio $L/D$, %(i.e., the larger is the dimensionality reduction factor), 
the more is the advantage of an optimal JPLT protocol over the best of the other two schemes. 
For instance, for $D=250$, the rate of an optimal JPLT protocol is about $10\%$ and $15\%$ more than that of the PIR-based scheme for $L/D=0.6$ and $L/D=0.4$, respectively.}
\end{remark}

\section{Linear JPLT Protocols and Linear Codes}\label{sec:LINEAR}
%The individual privacy and recoverability conditions impose a necessary but not sufficient condition for any IPLT protocol. %, stated in Lemma~\ref{lem:NCIPLT}. 

While any linear or non-linear JPLT protocol must satisfy the joint privacy and recoverability conditions, for linear JPLT protocols these conditions can be translated into the language of linear codes as discussed below. 
%particularly have an interesting coding-theoretic interpretation for linear JPLT protocols.  

%While the individual privacy and recoverability conditions must hold for any linear or non-linear JPLT protocol, they establish an interesting connection between linear JPLT protocols and linear codes. 
%Below, we discuss this connection for both deterministic and randomized protocols. 
  
%yields a $[D,L]$ code whose generator matrix is $\mathrm{V}_i$.
%For a coordinate $i$, we say that a collection $\mathrm{W}_i$ of $K-D$ coordinates is \emph{valid} if $\mathrm{W}_i$ does not contain $i$, and a collection $\mathrm{C}_i$ of $L$ codewords is \emph{valid} if their support does not contain $\mathrm{W}_i$. 
%For valid $\mathrm{W}_i$ and $\mathrm{C}_i$, we say that the pair $(\mathrm{W}_i,\mathrm{C}_i)$ is \emph{feasible} if puncturing the subcode generated by $\mathrm{C}_i$ at $\mathrm{W}_i$ yields a $[D,L]$ MDS code.\footnote{To puncture a linear code at a coordinate, the column corresponding to that coordinate is deleted from the generator matrix of the code.} 

%$\mathrm{W}$ and $\mathrm{V}$ are the demand's support and coefficient matrix, respectively.

%Consider a deterministic linear IPLT protocol. 
In the following, we refer to a JPLT-I or JPLT-II protocol, simply as a JPLT protocol, and denote both $\mathbbmss{V}_{I}$ for Model~I and $\mathbbmss{V}_{I\hspace{-0.04cm}I}$ for Model~II by $\mathbbmss{V}$ for the ease of notation.

Let $w\triangleq |\mathbbmss{W}|$ and $v\triangleq |\mathbbmss{V}|$, and let $\{\mathrm{W}_k\}_{k\in [w]}$ and $\{\mathrm{V}_l\}_{l\in [v]}$ be an arbitrary ordering of all elements in $\mathbbmss{W}$ and $\mathbbmss{V}$, respectively.
%Note that $|\mathbbmss{W}\times \mathbbmss{V}|=wv$. 
Consider an arbitrary linear JPLT protocol. 
For any instance $(\mathrm{W}_k,\mathrm{V}_l)$ for $k\in [w]$ and $l\in [v]$, the protocol can be specified by an ensemble of $n$ ($=n(k,l)$) distinct linear codes $\mathscr{C}^{1}_{k,l},\dots,\mathscr{C}^{n}_{k,l}$ of length $K$, for some integer $n$, and their respective probabilities $p^{1}_{k,l},\dots,p^{n}_{k,l}>0$. 
More specifically, for each $h\in [n]$, $\mathscr{C}^{h}_{k,l}$ is chosen with probability $p^{h}_{k,l}$ as the corresponding code for the instance $(\mathrm{W}_k,\mathrm{V}_l)$, i.e., the code corresponding to the coefficient matrix of the linear combinations that constitute the answer $\mathrm{A}^{[\mathrm{W}_k,\mathrm{V}_l]}$ to the query $\mathrm{Q}^{[\mathrm{W}_k,\mathrm{V}_l]}$.
%That is, with probability $p^{h}_{k,l}$, $\mathscr{C}^{h}_{k,l}$ is the code corresponding to the coefficient matrix of the linear combinations that constitute the answer $\mathrm{A}^{[\mathrm{W}_k,\mathrm{V}_l]}$ to the query $\mathrm{Q}^{[\mathrm{W}_k,\mathrm{V}_l]}$. 
Note that ${\sum_{h=1}^{n}p^{h}_{k,l} = 1}$. 
%Let $\{\mathscr{C}_1,\dots,\mathscr{C}_m\}$ be the underlying set of the multiset $\{\mathscr{C}^{n}_{k,l}\}_{k,l,n}$, formed by its distinct elements. 

%For any $(k,l)\in [w]\times[v]$, we denote by $\mathscr{C}_{k,l}$ the corresponding linear code for the instance $(\mathrm{W}_k,\mathrm{V}_l)$. 
%That is, $\mathscr{C}_{k,l}$ is the code corresponding to the coefficient matrix of the linear combinations that constitute the answer $\mathrm{A}^{[\mathrm{W}_k,\mathrm{V}_l]}$ to the query $\mathrm{Q}^{[\mathrm{W}_k,\mathrm{V}_l]}$. 
%Note that $\mathscr{C}_{k,l}$'s are not necessarily distinct, and $\{\mathscr{C}_{k,l}\}_{k,l}$ is a multiset in general. 
%Let $m$ be the number of distinct elements in the multiset $\{\mathscr{C}_{k,l}\}_{k,l}$, denoted by $\mathscr{C}_1,\dots,\mathscr{C}_m$. 
%Let $m$ be the number of distinct elements in the multiset $\{\mathscr{C}_{k,l}\}_{k,l}$, and let $\mathscr{C}_1,\dots,\mathscr{C}_m$ and $r_1,\dots,r_m$ be the %underlying set of the multiset $\{\mathscr{C}_{k,l}\}_{k,l}$, formed by its distinct elements.
%distinct elements and their multiplicities in the multiset $\{\mathscr{C}_{k,l}\}_{k,l}$, respetively.  

%for any given $\mathrm{W},\mathrm{V}$, and its corresponding linear code $\mathscr{C}$ generated by the answer's  coefficient matrix (i.e., the coefficient matrix of the linear combinations that constitute the server's answer).

Below, we introduce the notion of $(k,l)$-feasibility, which we will use to restate the joint privacy and recoverability conditions in the terminology of linear codes. 
% is a necessary and sufficient condition for recoverability, for the instance $(\mathrm{W}_k,\mathrm{V}_l)$.
For any $k,l$, we say that a linear code $\mathscr{C}$ of length $K$ is \emph{$(k,l)$-feasible} %for the coordinate $i$ 
%Given a linear code $\mathscr{C}$ of length $K$ and a coordinate ${i\in [K]}$, a pair ${(\tilde{\mathrm{W}},\tilde{\mathrm{V}})\in \mathbbmss{W}\times \mathbbmss{V}}$ is said to be \emph{$\mathscr{C}$-feasible} for the coordinate $i$ %with respect to $\mathscr{C}$ 
%if ${\mathrm{W}_{k}}$ contains $i$, and
if $\mathscr{C}$ contains a collection $\mathrm{C}$ of $L$ codewords whose support is a subset of $\mathrm{W}_{k}$, and the %subcode of $\mathscr{C}$
code generated by $\mathrm{C}$, when punctured at the coordinates indexed by $\mathrm{W}_{k}$, is identical to the code generated by $\mathrm{V}_{l}$.\footnote{Puncturing a linear code at a coordinate is performed by deleting the column pertaining to that coordinate from the generator matrix of the code.} 
Note that, for satisfying the recoverability condition, it is necessary and sufficient that for any $k,l,h$, the code $\mathscr{C}^{h}_{k,l}$ is $(k,l)$-feasible.

Note that $\{\mathscr{C}^{h}_{k,l}\}_{k,l,h}$ is a multiset in general because $\mathscr{C}^{h}_{k,l}$'s are not necessarily distinct. 
Let $m$ be the number of distinct elements in $\{\mathscr{C}^{h}_{k,l}\}_{k,l,h}$, and let $\mathscr{C}_1,\dots,\mathscr{C}_m$ be the distinct elements in $\{\mathscr{C}^{h}_{k,l}\}_{k,l,h}$.
For any $k\in [w]$ and $j\in [m]$, let $q_{k,j}$ be the sum of probabilities $p^{h}_{k,l}$ over all $l,h$ such that $\mathscr{C}^{h}_{k,l}$ is $(k,l)$-feasible, and $\mathscr{C}^{h}_{k,l}$ and $\mathscr{C}_{j}$ are identical. 
For any $j\in [m]$, let $r_j$ be the sum of probabilities $p^{h}_{k,l}$ over all $k,l,h$ such that $\mathscr{C}^{h}_{k,l}$ and $\mathscr{C}_{j}$ are identical.
Note that $q_{k,j}/r_j$ is the conditional probability that the message index set $\mathrm{W}_k$ is the demand's support, given that $\mathscr{C}_j$ is the code corresponding to the answer. 
% let $\{\mathscr{C}_1,\dots,\mathscr{C}_m\}$ be those codes $\mathscr{C}_{k,l}$ that are distinct.  
%By Lemma~\ref{lem:NCIPLT}, $n_{i,j}\geq 1$ for all $i,j$.
%For each coordinate $i$ and each code index $j$, suppose there exist $n_{i,j}$ $\mathscr{C}_j$-feasible pairs. 
It should be obvious that $q_{k,j}> 0$ for all $k,j$ is a necessary condition for joint privacy. 
Note that this condition is only necessary, and not sufficient. 
A necessary and sufficient condition for joint privacy is that for any $j\in [m]$, $q_{k,j}=q_j$ for all $k\in [w]$, for some $q_j>0$.
%This is obviously a stronger condition than the necessary (but not sufficient) condition mentioned earlier. 
%However, this necessary condition is less combinatorial, and has proven more instrumental in the proofs. 

%the necessary (but not sufficient) condition mentioned earlier is less combinatorial than this necessary and sufficient condition

%In contrast to this result which is more information theoretic, the necessary and sufficient condition for joint privacy in Section~\ref{sec:LINEAR} for linear protocols is more combinatorial.
%Note that this is different from $p_{i,j} = p$ for all $i$ and all $j$, for some $p>0$. 
%That is, the sum of probabilities for all coordinates do not need to be the same for a given code in the ensemble; instead, they must be the same over all codes in the ensemble. 
%Also, a necessary and sufficient condition for recoverability is that for any $(k,l,h)$,  $\mathscr{C}^{h}_{k,l}$ is $(k,l)$-feasible. % for all coordinates $i\in \mathrm{W}$, where $(\mathrm{W}_{k^{*}},\mathrm{V}_{l^{*}}) = (\mathrm{W},\mathrm{V})$.

For any $k,l$, let $d_{k,l}$ be the expected value of the dimension of a randomly chosen code from the ensemble $\{\mathscr{C}^{1}_{k,l},\dots,\mathscr{C}^{n}_{k,l}\}$ for the instance $(\mathrm{W}_k,\mathrm{V}_l)$, according to the probability distribution $\{p^{1}_{k,l},\dots,p^{n}_{k,l}\}$. 
%That is, $d_{k,l}$ is the weighted average of the dimension of the codes $\mathscr{C}^{1}_{k,l},\dots,\mathscr{C}^{n}_{k,l}$, where the weights are specified by the probabilities $p^{1}_{k,l},\dots,p^{n}_{k,l}$. 
Let $d_{\text{ave}}$ be the average of $d_{k,l}$'s over all $k,l$. 
%The rate of a randomized linear IPLT protocol is equal to $1/d_{\text{ave}}$. 
It should be obvious that the rate of a linear JPLT protocol is equal to $1/d_{\text{ave}}$. 
Maximizing the rate of a linear JPLT protocol is then equivalent to minimizing $d_{\text{ave}}$, subject to the aforementioned necessary and sufficient conditions for joint privacy and recoverability. 

%Notwithstanding that for both deterministic and randomized linear protocols these necessary and sufficient conditions are stronger than the result of Lemma~\ref{lem:NCIPLT}, the latter is more information-theoretic and less combinatorial than the former, and has proven more useful in the converse proof. 
%%%%%%%%%%%%%%%%%%%%%%%%%%%%

%The rate of a linear JPLT protocol is then given by $1/d_{\text{ave}}$, where $d_{\text{ave}}$ denotes the average of $d_{k,l}$'s over all pairs $(k,l)$, and $d_{k,l}$ is the expected value of the dimensions of the codes $\mathscr{C}^{1}_{k,l},\dots,\mathscr{C}^{n}_{k,l}$ according to the probability distribution specified by $p^{1}_{k,l},\dots,p^{n}_{k,l}$. The corresponding capacity of a JPLT protocol is then equivalent to minimizing $1/d_{\text{ave}}$, subject to the necessary and sufficient conditions described above, noting that capacity is the maximum rate over all JPLT protocols.    

%%%%%%%%%%%%%%%%%%%%%%%%%%%%

\section{Proof of Theorem~\ref{thm:JPLT1}}\label{sec:JPLT1}
We prove the converse in Section~\ref{subsec:JPLT1-Conv}, and present the achievability scheme in Section~\ref{subsec:JPLT1-Ach}.

\subsection{Converse Proof}\label{subsec:JPLT1-Conv}
The following result is useful in the proof of converse for the JPLT-I problem. 
%The following lemma states a necessary %(yet not always sufficient) 
%condition for any JPLT-I protocol. %This result follows immediately from the joint privacy and recoverability conditions, and its proof is omitted for brevity. 

\begin{lemma}\label{lem:NCJPLT1}
Given any JPLT-I protocol, for any $\tilde{\mathrm{W}}\in\mathbbm{W}$, there must exist $\tilde{\mathrm{V}}\in\mathbbm{V}_{I}$ such that \[H(\mathbf{Z}^{[\tilde{\mathrm{W}},\tilde{\mathrm{V}}]}| \mathbf{A}, \mathbf{Q})= 0.\] 		
\end{lemma}

\begin{proof}
The proof is by the way of contradiction. Consider an arbitrary JPLT-I protocol. 
Let $\mathrm{Q}$ and $\mathrm{A}$ be the query and the corresponding answer generated by this protocol for an arbitrary instance $(\mathrm{W},\mathrm{V})$. 
Consider an arbitrary ${\tilde{\mathrm{W}}\in\mathbbm{W}}$. 
Suppose that there does not exist $\tilde{\mathrm{V}}\in\mathbbm{V}_{I}$ such that ${H(\mathbf{Z}^{[\tilde{\mathrm{W}},\tilde{\mathrm{V}}]}| \mathbf{A}, \mathbf{Q})= 0}$. 
This implies that $\mathbf{W}\neq \tilde{\mathrm{W}}$, given that $\mathbf{Q}=\mathrm{Q}$ %, i.e., $\tilde{\mathrm{W}}$ cannot be the support of the demand 
(otherwise, if $\mathbf{W}=\tilde{\mathrm{W}}$, the
%then the user could not retrieve the demand (i.e., 
recoverability condition is not satisfied). %Hence, from the perspective of the server, $\mathrm{W}^{*}$ could not be a potential support index set of the demand.
Thus, ${\Pr(\mathbf{W}=\tilde{\mathrm{W}}|\mathbf{Q}=\mathrm{Q}) = 0}$. 
This is, however, a contradiction because by the joint privacy condition,  ${\Pr(\mathbf{W}=\tilde{\mathrm{W}}|\mathbf{Q}=\mathrm{Q})}={\Pr(\mathbf{W}=\tilde{\mathrm{W}})}=1/C_{K,D}\neq 0$. 
%given the query $\mathrm{Q}$, every $D$-subset of message indices, including $\mathrm{W}^{*}$, must be equally likely to be the demand's support index set.            
\end{proof}

When considering linear protocols, the result of Lemma~\ref{lem:NCJPLT1} is equivalent to the necessary (but not sufficient) condition for joint privacy in Section~\ref{sec:LINEAR}.  
In contrast to this result which is more information theoretic and more instrumental in the proofs, the necessary and sufficient condition for joint privacy in Section~\ref{sec:LINEAR} is more combinatorial and harder to analyze.
Moreover, the necessary and sufficient condition for joint privacy in Section~\ref{sec:LINEAR} is specific to linear protocols; whereas 
Lemma~\ref{lem:NCJPLT1} applies also to non-linear protocols.

\begin{lemma}\label{lem:JPLT1-Conv}
The rate of any JPLT-I protocol for $K$ messages, demand's support size $D$, and demand's dimension $L$ is upper bounded by $L/(K-D+L)$.
\end{lemma}

\begin{proof}
Consider an arbitrary JPLT-I protocol that generates a query-answer pair $(\mathrm{Q}^{[\mathrm{W},\mathrm{V}]},\mathrm{A}^{[\mathrm{W},\mathrm{V}]})$ for any given $(\mathrm{W},\mathrm{V})$. 
For simplifying the notation, we denote the random variables $\mathbf{Q}^{[\mathbf{W},\mathbf{V}]}$ and $\mathbf{A}^{[\mathbf{W},\mathbf{V}]}$ by $\mathbf{Q}$ and $\mathbf{A}$, respectively. 
To show that the rate is upper bounded by $L/(K-D+L)$, we need to show that ${H(\mathbf{A})\geq (K-D+L)B}$, where ${B = N\log_2 q}$ is the entropy of a uniformly distributed message over $\mathbbmss{F}_q^{N}$.
%For the ease of notation, we define 

Let $T\triangleq K-D+1$. 
For each $i\in [T]$, let $\mathrm{W}_i \triangleq {\{i,i+1,\dots,i+D-1\}}$. 
Note that ${\mathrm{W}_1,\dots,\mathrm{W}_T\in \mathbbm{W}}$.
By Lemma~\ref{lem:NCJPLT1}, there exists ${\mathrm{V}_i\in \mathbbm{V}_{I}}$ for ${i\in [T]}$ such that ${H(\mathbf{Z}_i |\mathbf{A},\mathbf{Q}) = 0}$, where ${\mathbf{Z}_i\triangleq \mathbf{Z}^{[\mathrm{W}_i,\mathrm{V}_i]}}$. (Note that $\mathrm{V}_i$ is an MDS matrix.)
This readily implies that $H(\mathbf{Z}_1,\dots,\mathbf{Z}_{T}|\mathbf{A},\mathbf{Q})=0$ since $H(\mathbf{Z}_1,\dots,\mathbf{Z}_{T}|\mathbf{A},\mathbf{Q})\leq \sum_{i=1}^{T} H(\mathbf{Z}_i|\mathbf{A},\mathbf{Q}) = 0$. 
Thus, 
\begin{align}
    H(\mathbf{A})&\geq H(\mathbf{A}|\mathbf{Q}) +H(\mathbf{Z}_1,\dots,\mathbf{Z}_{T}|\mathbf
    {Q},\mathbf{A}) \label{eq:1}\\
    &= H(\mathbf{Z}_1,\dots,\mathbf{Z}_{T}|\mathbf{Q})
    +H(\mathbf{A}|\mathbf{Q},\mathbf{Z}_1,\dots,\mathbf{Z}_{T}) \label{eq:2}\\
    &\geq H(\mathbf{Z}_1,\dots,\mathbf{Z}_{T}), \label{eq:3}
\end{align} 
where~\eqref{eq:1} holds because $H(\mathbf{Z}_1,\dots,\mathbf{Z}_{T}|\mathbf{A},\mathbf{Q})=0$, as shown earlier; \eqref{eq:2} follows from the chain rule of conditional entropy;
and \eqref{eq:3} holds because 
(i) $\mathbf{Z}_i$'s are independent from $\mathbf{Q}$, noting that $\mathbf{Z}_i$'s only depend on $\mathbf{X}$, and $\mathbf{Q}$ is independent of $\mathbf{X}$, and 
(ii) $H(\mathbf{A}|\mathbf{Q},\mathbf{Z}_1,\dots,\mathbf{Z}_{T})\geq 0$. 

To lower bound $H(\mathbf{Z}_1,\dots,\mathbf{Z}_{T})$, we proceed as follows. 
By the chain rule of entropy, we have %$H(\mathbf{Z}_1,\dots,\mathbf{Z}_{T})=H(\mathbf{Z}_1)+\sum_{i=2}^{T} H(\mathbf{Z}_{i}|\mathbf{Z}_{1},\dots,\mathbf{Z}_{i-1})$.
%\begin{comment}
\begin{equation}\label{eq:4}
H(\mathbf{Z}_1,\dots,\mathbf{Z}_{T})=H(\mathbf{Z}_1)+\sum_{i=2}^{T} H(\mathbf{Z}_{i}|\mathbf{Z}_{1},\dots,\mathbf{Z}_{i-1}).    
\end{equation}
%\end{comment}
Let $\mathbf{Z}_{i,1},\dots,\mathbf{Z}_{i,L}$ be the $L$ rows of the matrix $\mathbf{Z}_i$, i.e., $\mathbf{Z}_{i,l}\triangleq \mathrm{v}_{i,l} \mathbf{X}_{\mathrm{W}_i}$, where $\mathrm{v}_{i,l}$ is the $l$th row of $\mathrm{V}_i$. 
Note that $\mathbf{Z}_i$ consists of $L$ row-vectors $\mathbf{Z}_{i,1},\dots,\mathbf{Z}_{i,L}$, and 
these vectors are independent because their corresponding coefficient vectors $\mathrm{v}_{i,1},\dots,\mathrm{v}_{i,L}$ are linearly independent. 
Moreover, $\mathbf{Z}_{i,1},\dots,\mathbf{Z}_{i,L}$ are uniform over $\mathbbmss{F}_q^{N}$, i.e., $H(\mathbf{Z}_{i,l})=B$ for $l\in [L]$. 
Thus, $H(\mathbf{Z}_i)=H(\mathbf{Z}_{i,1},\dots,\mathbf{Z}_{i,L})=LB$, particularly, ${H(\mathbf{Z}_1)=LB}$. 
Note, also, that there exists some ${l\in [L]}$ such that $\mathbf{Z}_{i,l}$ is dependent on $\mathbf{X}_{i+D-1}$, i.e., the coefficient of $\mathbf{X}_{i+D-1}$ in the linear combination $\mathbf{Z}_{i,l}$ is nonzero. 
Otherwise, $\mathrm{V}_i$ contains an all-zero column, which contradicts with the fact that $\mathrm{V}_i$ is MDS. 
Moreover, there does not exist any $l\in [L]$ such that $\mathbf{Z}_{j,l}$ for any $j<i$ depends on $\mathbf{X}_{i+D-1}$ (by construction of $\mathrm{W}_1,\dots,\mathrm{W}_i$). 
%$\mathbf{X}_{i+D-1}$ does not belong to the support set of any of the components of $\mathbf{Z}_j$ for any $j<i$ (by construction). 
This implies that there exists at least one row-vector, namely, $\mathbf{Z}_{i,l}$, that is independent of 
%cannot be written as a linear combination of 
the row-vectors pertaining to $\mathbf{Z}_1,\dots,\mathbf{Z}_{i-1}$. % , for every ${i\in \mathcal{T}\setminus \{1\}}$,
%Thus, $\mathbf{Z}_{i,l}$ is independent of $\mathbf{Z}_1,\dots,\mathbf{Z}_{i-1}$. 
This further implies that $H(\mathbf{Z}_{i}|\mathbf{Z}_{1},\dots,\mathbf{Z}_{i-1})\geq H(\mathbf{Z}_{i,l})=B$, and consequently, $\sum_{i=2}^{T} H(\mathbf{Z}_{i}|\mathbf{Z}_{1},\dots,\mathbf{Z}_{i-1})\geq {(T-1)B}$. 
From~\eqref{eq:4}, it then follows that 
\begin{equation}\label{eq:5}
H(\mathbf{Z}_1,\dots,\mathbf{Z}_T)\geq LB + (T-1)B = (K-D+L)B.    
\end{equation} Combining~\eqref{eq:3} and~\eqref{eq:5}, we have $H(\mathbf{A})\geq (K-D+L)B$.%, as was to be shown.
\end{proof}

\subsection{Achievability Scheme}\label{subsec:JPLT1-Ach}
In this section, we present a JPLT-I protocol, termed the \emph{Specialized MDS Code protocol}, which is capacity-achieving for sufficiently large $q$---depending on the parameters $K,D,L$.
An illustrative example of this protocol can be found in Appendix~\ref{app:1}. 

The Specialized MDS Code protocol consists of three steps as described below.\vspace{0.125cm} 

\textbf{Step 1:} Given the demand's support $\mathrm{W}\in \mathbbm{W}$ and the demand's coefficient matrix $\mathrm{V}=[\mathrm{v}_1^{\transpose},\dots,\mathrm{v}_l^{\transpose}]^{\transpose}\in \mathbbm{V}_{I}$, the user constructs a query $\mathrm{Q}^{[\mathrm{W},\mathrm{V}]}$ in the form of a matrix $\mathrm{G}$, such that the user's query, i.e., the matrix $\mathrm{G}$, and the server's corresponding answer $\mathrm{A}^{[\mathrm{W},\mathrm{V}]}$, i.e., the matrix $\mathrm{Y}=\mathrm{G}\mathrm{X}$, 
satisfy the recoverability and joint privacy conditions.

To satisfy the joint privacy condition, it is required that, for any index set $\tilde{\mathrm{W}}\in\mathbbm{W}$, the code generated by the matrix $\mathrm{G}$ contains $L$ codewords whose support are some subsets of $\tilde{\mathrm{W}}$, and the coordinates of these codewords (indexed by $\tilde{\mathrm{W}}$) form an MDS matrix $\tilde{\mathrm{V}}\in \mathbbm{V}_{I}$. 
By the properties of MDS codes~\cite{R2006}, it is easy to verify that the generator matrix of any $[K,K-D+L]$ MDS code satisfies this requirement. 
However, not any such generator matrix is guaranteed to satisfy the recoverability condition.
For satisfying the recoverability condition, it is required that $\mathrm{G}$, as a generator matrix, generates a code that contains $L$ codewords with the support $\mathrm{W}$, and the coordinates of these codewords (indexed by $\mathrm{W}$) must conform to the coefficient matrix $\mathrm{V}$. 
To construct a matrix $\mathrm{G}$ that satisfies these requirements, the user proceeds as follows.

First, the user constructs the parity-check matrix $\myLambda$ of the $[D,L]$ MDS code generated by $\mathrm{V}$. 
Since $\mathrm{V}$ is an MDS matrix, then $\myLambda$ generates a $[D,D-L]$ MDS code. % (i.e., the dual of the MDS code generated by $\mathrm{V}$). 
The user then constructs a ${(D-L)\times K}$ matrix $\mathrm{H}$ that satisfies the following two conditions: (i) the matrix $\mathrm{H}$ contains $\myLambda$ as a submatrix, and (ii) the matrix $\mathrm{H}$ is MDS. 
%(i) $\myLambda$ is the submatrix of $\mathrm{H}$ restricted to columns indexed by $\mathrm{W}$, and 
%(ii) $\mathrm{H}$ is an MDS matrix. 
Since $\myLambda$ is an MDS matrix, constructing $\mathrm{H}$ reduces to extending the $[D,D-L]$ MDS code generated by $\myLambda$ to a $[K,D-L]$ MDS code.
(%The feasibility of such an extension depends on the field size $q$. 
An application of Schwartz-Zippel lemma shows that such an extension is feasible so long as $q$ is sufficiently large.) 
The user then constructs a matrix $\tilde{\mathrm{H}}$ by permuting the columns of $\mathrm{H}$ arbitrarily such that $\myLambda$ is the submatrix of $\tilde{\mathrm{H}}$ restricted to the columns indexed by $\mathrm{W}$. 
%With a slight abuse of notation, 
For simplicity, we also denote $\tilde{\mathrm{H}}$ by $\mathrm{H}$.  
Next, the user constructs a $(K-D+L)\times K$ matrix $\mathrm{G}$ that generates the MDS code defined by the parity-check matrix $\mathrm{H}$. 
(Since $\mathrm{H}$ generates a $[K,D-L]$ MDS code, $\mathrm{H}$ is the parity-check matrix of a $[K,K-D+L]$ MDS code.)
The user then sends $\mathrm{G}$ as the query $\mathrm{Q}^{[\mathrm{W},\mathrm{V}]}$ to the server.\vspace{0.125cm} 

\textbf{Step 2:} Given the query $\mathrm{Q}^{[\mathrm{W},\mathrm{V}]}$, i.e., the matrix $\mathrm{G}$, the server computes the $(K-D+L)\times N$ matrix $\mathrm{Y}\triangleq \mathrm{G}\mathrm{X}$, and sends $\mathrm{Y}$ as the answer $\mathrm{A}^{[\mathrm{W},\mathrm{V}]}$ back to the user.\vspace{0.125cm}  

\textbf{Step 3:} Upon receiving the answer $\mathrm{A}^{[\mathrm{W},\mathrm{V}]}$, i.e., the matrix $\mathrm{Y}$, 
the user constructs a matrix $[\tilde{\mathrm{G}},\tilde{\mathrm{Y}}]$ by performing row operations on the augmented matrix $[\mathrm{G},\mathrm{Y}]$, so as to zero out the submatrix formed by the first $L$ rows and the columns indexed by $[K]\setminus \mathrm{W}$.
Since the submatrix of $[\tilde{\mathrm{G}},\tilde{\mathrm{Y}}]$ formed by the first $L$ rows and the columns indexed by $\mathrm{W}$ (or $[K]\setminus\mathrm{W}$) is equal to the matrix $\mathrm{V}$ (or an all-zero matrix), 
the $l$th row of the demand matrix $\mathrm{Z}^{[\mathrm{W},\mathrm{V}]}$, i.e., $\mathrm{v}_l\mathrm{X}_{\mathrm{W}}$, for ${l\in [L]}$, can be recovered from the $l$th row of the matrix $\tilde{\mathrm{Y}}$.

In the following, we provide a more explicit description of the Specialized MDS Code protocol for the cases in which the coefficient matrix $\mathrm{V}$ generates a GRS code. 
We refer to this protocol as the \emph{Specialized GRS Code protocol}. 
Note that this protocol is applicable for any field size $q\geq K$.\vspace{0.125cm}  

\textbf{Step 1:} Suppose that 
%We now describe how to explicitly construct the matrix $\mathrm{G}$ when the coefficient matrix $\mathrm{V}$ generates a GRS code, i.e., 
the entry $(i,j)$ of $\mathrm{V}$ is given by $\mathrm{V}_{i,j} \triangleq v_{j} \omega_{j}^{i-1}$, where $v_{1},\dots,v_{D}$ are $D$ elements from $\mathbbmss{F}_q\setminus \{0\}$, and $\omega_{1},\dots,\omega_{D}$ are $D$ distinct elements from $\mathbbmss{F}_q$. 
%Let $\mathcal{D} \triangleq \{1,\dots,D\}$. 
The parameters $v_{1},\dots,v_{D}$ and $\omega_{1},\dots,\omega_{D}$ are the multipliers and the evaluation points of the GRS code generated by $\mathrm{V}$, respectively. 
Since the dual of a GRS code is also a GRS code~\cite{R2006}, the parity-check matrix $\myLambda$ of the GRS code generated by $\mathrm{V}$ is a ${(D-L)\times D}$ matrix whose entry $(i,j)$ is given by ${\myLambda_{i,j} \triangleq \lambda_{j} \omega_{j}^{i-1}}$, where \[{\lambda_{j} \triangleq v_{j}^{-1}\prod_{k\in [D]\setminus \{j\}} (\omega_{j}-\omega_{k})^{-1}}.\]
Note that $\lambda_1,\dots,\lambda_D$ are nonzero.  
Extending the $(D-L)\times D$ matrix $\myLambda$ to a $(D-L)\times K$ matrix $\mathrm{H}$---satisfying the conditions (i) and (ii)---is performed as follows. 

Let $\mathrm{W}=\{i_1,\dots,i_D\}$ and $[K]\setminus \mathrm{W} = \{i_{D+1},\dots,i_K\}$, and 
let $\pi$ be a permutation on $[K]$ such that $\pi(j) = i_j$. 
Let $\lambda_{D+1},\dots,\lambda_{K}$ be $K-D$ elements  chosen randomly (with replacement) from $\mathbbmss{F}_p\setminus \{0\}$, and 
let $\omega_{D+1},\dots,\omega_K$ be $K-D$ elements chosen randomly (without replacement) from ${\mathbbmss{F}_p\setminus \{\omega_1,\dots,\omega_D\}}$.
For every ${j\in [D]}$, let the $\pi(j)$th column of $\mathrm{H}$ be the $j$th column of $\myLambda$, and for every $j\in [K]\setminus [D]$, let the $\pi(j)$th column of $\mathrm{H}$ be $[\lambda_j,\lambda_{j}\omega_j,\dots,\lambda_j\omega_j^{D-L-1}]^{\transpose}$. 
Since $\mathrm{H}$ is the parity-check matrix of a ${[K,K-D+L]}$ GRS code, the generator matrix of this code, $\mathrm{G}$, can be constructed by taking the $\pi(j)$th column of $\mathrm{G}$ to be $[\alpha_j,\alpha_j\omega_j,\dots,\alpha_j\omega_j^{K-D+L-1}]^{\transpose}$, where \[\alpha_j\triangleq \lambda_{j}^{-1}\prod_{k\in [K]\setminus \{j\}} (\omega_{j}-\omega_{k})^{-1}.\]
The parameters $\{\alpha_{j}\}_{j\in [K]}$ and $\{\omega_{j}\}_{j\in [K]}$ are the multipliers and the evaluation points of the GRS code generated by $\mathrm{G}$, respectively. The user then sends the matrix $\mathrm{G}$ to the server.\vspace{0.125cm}  

\textbf{Step 2:} Given the matrix $\mathrm{G}$, the server computes the matrix $\mathrm{Y}=\mathrm{G}\mathrm{X}$, where the $i$th row of $\mathrm{Y} =  [Y_1^{\transpose},\dots,Y_{K-D+L}^{\transpose}]^{\transpose}$ is given by \[Y_i \triangleq \sum_{j=1}^{K}\alpha_{j}\omega_j^{i-1}X_{j},\] and sends $\mathrm{Y}$ back to the user.\vspace{0.125cm}  

\textbf{Step 3:} Given the matrix $\mathrm{Y}$, the user recovers the demand matrix $\mathrm{Z}^{[\mathrm{W},\mathrm{V}]}$ as follows. 
%When $\mathrm{V}$ generates a GRS code, $\mathrm{Z}^{[\mathrm{W},\mathrm{V}]}$ can be recovered from the vector $\mathrm{y}$ as follows. 
First, the user constructs $L$ polynomials $f_1(x),\dots,f_L(x)$, where \[f_{l}(x)\triangleq {x^{l-1}\prod_{j=D+1}^{K} (x-\omega_j)}.\] 
For each $l\in [L]$, let $\mathrm{c}_{l}\triangleq [c_{l,1},\dots,c_{l,K-D+L}]^{\transpose}$, where $c_{l,i}$ is the coefficient of the monomial $x^{i-1}$ in the polynomial expansion of $f_{l}(x)$. 
The user then recovers the $l$th row of the demand matrix $\mathrm{Z}^{[\mathrm{W},\mathrm{V}]}$, namely, 
$\mathrm{v}_l\mathrm{X}_{\mathrm{W}}$, by computing $\mathrm{c}_{l}^{\transpose}\mathrm{Y}$.

\begin{proposition}[Symmetry Property of MDS Codes]\label{prop:MDSSymmetry}
Given any $[n,k]$ MDS code, for any ${\mathrm{S}\subseteq [n]}$ such that ${|\mathrm{S}|\geq n-k+1}$, the code space contains a unique ${(|\mathrm{S}|-n+k)}$-dimensional subspace on the coordinates indexed by $\mathrm{S}$, and any basis of this subspace (restricted to the coordinates indexed by $\mathrm{S}$) forms an MDS matrix.
\end{proposition}

\begin{proof}
Consider an arbitrary $[n,k]$ MDS code $\mathscr{C}$. 
Let $d\triangleq n-k+1$ be the minimum distance of $\mathscr{C}$. 
By the properties of MDS codes~\cite{R2006}, for any $d$-subset $\mathrm{T}\subseteq [n]$, the code $\mathscr{C}$ has a codeword whose support is $\mathrm{T}$.
Consider an arbitrary $\mathrm{S}\subseteq [n]$ such that $|\mathrm{S}|\geq d$. 
Let $s\triangleq |\mathrm{S}|$, and $\mathrm{S} \triangleq \{l_1,\dots,l_{s}\}$. 
Let $m\triangleq s-d+1$.
Note that $m\leq s$. 
For each ${i\in [m]}$, let $\mathrm{S}_i \triangleq \{l_{i},\dots,l_{i+d-1}\}$, and 
let $\mathrm{c}_i$ be a codeword of $\mathscr{C}$ whose support is $\mathrm{S}_i$. 
Note that $\mathrm{c}_i$'s are row-vectors of length $n$. 
Consider an $m\times n$ matrix $\mathrm{C}$ whose $i$th row is $\mathrm{c}_i$, i.e., $\mathrm{C}\triangleq [\mathrm{c}_1^{\transpose},\dots,\mathrm{c}_{m}^{\transpose}]^{\transpose}$. 
Note that the $l_{i+d-1}$th entry of the $i$th row of $\mathrm{C}$ is nonzero for each $i\in [m]$, and the $l_{i+d-1}$th entry of the $j$th row of $\mathrm{C}$ is zero for any $j<i$. 
This readily implies that $\mathrm{rank}(\mathrm{C})=m$. Thus, the row space of $\mathrm{C}$, i.e., the (linear) span of the codewords $\mathrm{c}_1,\dots,\mathrm{c}_m$, is an $m$-dimensional subspace on the coordinates indexed by $\mathrm{S}=\cup_{i=1}^{m} \mathrm{S}_i$. 
Note that $m={s-d+1}={s-(n-k+1)+1}={s-n+k}$.
This proves that the code space contains an $(s-n+k)$-dimensional subspace on the coordinates indexed by $\mathrm{S}$. 

Next, we show that any basis of the subspace spanned by the rows of $\mathrm{C}$ (restricted to the coordinates indexed by $\mathrm{S}$) forms an MDS matrix. 
Consider an arbitrary basis of this subspace. 
The matrix formed by this basis can be written as $\mathrm{R}\mathrm{C}$ for some $m\times m$ invertible matrix $\mathrm{R}$. 
Let $\hat{\mathrm{C}}$ be an $m\times s$ submatrix of $\mathrm{C}$ formed by the columns indexed by $\mathrm{S}$.
We need to show that $\mathrm{R}\hat{\mathrm{C}}$ is an MDS matrix. 
If $\hat{\mathrm{C}}$ is an MDS matrix, any $m\times m$ submatrix of $\hat{\mathrm{C}}$, and consequently, any $m\times m$ submatrix of $\mathrm{R}\hat{\mathrm{C}}$, is invertible, and hence, $\mathrm{R}\hat{\mathrm{C}}$ is an MDS matrix. 
Thus, it suffices to show that $\hat{\mathrm{C}}$ is an MDS matrix. 
Consider the $[s,m]$ code $\hat{\mathscr{C}}$ generated by $\hat{\mathrm{C}}$. 
The minimum distance of $\hat{\mathscr{C}}$ is at most ${s-m+1=n-k+1}$ ($=d$). 
The weight of the codewords of $\hat{\mathscr{C}}$ corresponding to the rows of $\hat{\mathrm{C}}$ is $d$. 
Moreover, any other (nonzero) codeword of $\hat{\mathscr{C}}$ is a linear combination of the rows of $\hat{\mathrm{C}}$, and has a weight at least $d$.  
(If $\hat{\mathscr{C}}$ has a codeword of weight less than $d$, then $\mathscr{C}$ must have a codeword of weight less than $d$, which is a contradiction since the minimum distance of $\mathscr{C}$ is $d$.) 
Thus, the minimum distance of $\hat{\mathscr{C}}$ is $d$ ($=s-m+1$), and $\hat{\mathscr{C}}$ is an $[s,m]$ MDS code.%, as was to be shown. 

Now, we prove the uniqueness by the way of contradiction.
Suppose that the code space contains two distinct subspaces on the coordinates indexed by $\mathrm{S}$. 
For $i\in \{1,2\}$, let $\mathrm{M}_i$ be an $m\times n$ matrix formed by an arbitrary basis of the $i$th subspace. 
Consider the matrix $\mathrm{M}=[\mathrm{M}_1^{\transpose},\mathrm{M}_2^{\transpose}]^{\transpose}$. 
Note that the rows of $\mathrm{M}$ are codewords of $\mathscr{C}$.
Obviously, $\tilde{m}\triangleq \mathrm{rank}(\mathrm{M})>m$. 
This is because $\mathrm{rank}(\mathrm{M}_1)=m$, and there exists at least one row in $\mathrm{M}_2$ that is linearly independent of the rows of $\mathrm{M}_1$. 
Let $\tilde{\mathrm{M}}$ be an $\tilde{m}\times n$ matrix formed by an arbitrary basis of the row space of $\mathrm{M}$.
Note that $\mathrm{rank}(\tilde{\mathrm{M}})=\tilde{m}$. 
By performing Gauss-Jordan elimination on a properly chosen column-permutation of $\tilde{\mathrm{M}}$, 
we can obtain a matrix of the form $[\mathrm{I},\mathrm{P},0]$, where $\mathrm{I}$ is an $\tilde{m}\times\tilde{m}$ identity matrix, $\mathrm{P}$ is an $\tilde{m}\times (n-\tilde{m})$ matrix, and $0$ is an $\tilde{m}\times (n-s)$ all-zero matrix. 
Note that the row space of $[\mathrm{I},\mathrm{P},0]$ is the same as the row space of $\tilde{\mathrm{M}}$ which is itself the same as the row space of $\mathrm{M}$, and hence, the rows of $[\mathrm{I},\mathrm{P},0]$ are codewords of $\mathscr{C}$. 
Fix an arbitrary $i\in [\tilde{m}]$. 
Consider the codeword corresponding to the $i$th row of $[\mathrm{I},\mathrm{P},0]$. 
The weight of this codeword is at most $s-\tilde{m}+1$, because there is only one nonzero coordinate within the first $\tilde{m}$ coordinates, and there are at most $s-\tilde{m}$ nonzero coordinates within the last $n-\tilde{m}$ coordinates. 
Thus, the minimum distance of $\mathscr{C}$ is at most $s-\tilde{m}+1$ which is strictly less than $s-m+1=s-(s-n+k)+1 = n-k+1=d$ since $\tilde{m}>m$. 
This is a contradiction because the minimum distance of $\mathscr{C}$ is $d$. 
%This completes the proof of uniqueness. 
\end{proof}

\begin{lemma}\label{lem:JPLT1-Ach}
The Specialized MDS Code protocol is a JPLT-I protocol, and achieves the rate $L/(K-D+L)$. 
\end{lemma}

\begin{proof}
Since the answer $\mathrm{Y}=\mathrm{G}\mathrm{X}$ is a matrix with $K-D+L$ rows, and the rows of this matrix are linearly independent coded combinations of the messages $\mathbf{X}_1,\dots,\mathbf{X}_K$ (noting that the matrix $\mathrm{G}$ has full rank), the entropy of the answer is given by ${(K-D+L)B}$, where $B$ is the entropy of a message. 
Thus, the rate of this protocol is $L/(K-D+L)$. 

Next, we prove that the joint privacy condition is satisfied. 
Note that the matrix $\mathrm{G}$ generates a $[K,K-D+L]$ MDS code with minimum distance $D-L+1$. 
By the symmetry property of MDS codes (Proposition~\ref{prop:MDSSymmetry}), the row space of $\mathrm{G}$ contains a unique $L$-dimensional subspace on every $D$-subset of coordinates. 
Note that each of these $L$-dimensional subspaces (corresponding to a distinct $D$-subset of coordinates) %(each subspace pertaining to a distinct $D$-subset of coordinates) 
is equally likely to be the subspace spanned by the rows of the demand's global coefficient matrix, from the server's perspective. 
Combining these arguments, given the matrix $\mathrm{G}$, every $D$-subset of message indices is equally likely to be the demand's support. 
This completes the proof of joint privacy. %in the support set of the demand are jointly private.

The recoverability follows readily from the construction. 
Let ${\mathrm{U}}$ be the global coefficient matrix of the demand. %, where $\mathrm{u}_{l}$ is a row-vector of length $K$ such that $\mathrm{u}_{l}$ restricted to its components indexed by $\mathrm{W}$ is equal to the vector $\mathrm{v}_l$, and the rest of the components of $\mathrm{u}_{l}$ are all zero. 
%Note that $\mathrm{V}$ is a submatrix of $\mathrm{U}$ formed by the columns indexed by $\mathrm{W}$. 
%To prove recoverability, 
We need to show that %the rows of $\mathrm{U}$ are in the row space of $\mathrm{G}$, i.e., 
the rows of $\mathrm{U}$ are $L$ codewords of the code generated by $\mathrm{G}$. 
Since $\mathrm{H}$ is the parity-check matrix of the code generated by $\mathrm{G}$, this is equivalent to showing that $\mathrm{U}\mathrm{H}^{\transpose}$ is an all-zero matrix. %, and this holds because of the following two reasons. 
This can be shown as follows. 
Firstly, the submatrix of $\mathrm{U}\mathrm{H}^{\transpose}$ restricted to the columns indexed by $\mathrm{W}$ is equal to $\mathrm{V}\myLambda^{\transpose}$, and $\mathrm{V}\myLambda^{\transpose}$ is an all-zero matrix because $\myLambda$ is the parity-check matrix of the code generated by $\mathrm{V}$. 
Secondly, the submatrix of $\mathrm{U}\mathrm{H}^{\transpose}$ formed by the columns indexed by $[K]\setminus \mathrm{W}$ is an all-zero matrix because the submatrix of $\mathrm{U}$ restricted to these columns is an all-zero matrix. Thus, $\mathrm{U}\mathrm{H}^{\transpose}$ is an all-zero matrix. 
This completes the proof of recoverability.%, as was to be shown.
\end{proof}

\section{Proof of Theorem~\ref{thm:JPLT2}}\label{sec:JPLT2}
The proof of converse is given in Section~\ref{subsec:JPLT2-Conv}, and the achievability scheme is presented in Section~\ref{subsec:JPLT2-Ach}.

\subsection{Converse Proof}\label{subsec:JPLT2-Conv}
%Similar to Lemma~\ref{lem:NCJPLT1} for the JPLT-I protocols, 
The converse proof for the JPLT-II problem relies on the following result. 

% provides a necessary %(but not sufficient) 
%condition for any JPLT-II protocol.

%This result follows from the exact same line as in the proof of Lemma~\ref{lem:NCJPLT1}, except where $\mathbbm{V}_{I}$ is replaced by $\mathbbm{V}_{I\hspace{-0.04cm}I}$.
%The proof is omitted to avoid repetition. 

\begin{lemma}\label{lem:NCJPLT2}
Given any JPLT-II protocol, for any $\tilde{\mathrm{W}}\in\mathbbm{W}$, there must exist $\tilde{\mathrm{V}}\in\mathbbm{V}_{I\hspace{-0.04cm}I}$ such that \[H(\mathbf{Z}^{[\tilde{\mathrm{W}},\tilde{\mathrm{V}}]}| \mathbf{A}, \mathbf{Q})= 0.\] 		
\end{lemma}

\begin{proof}
The result follows from the same argument as in the proof of Lemma~\ref{lem:NCJPLT1}, except where $\mathbbm{V}_{I}$ is replaced by $\mathbbm{V}_{I\hspace{-0.04cm}I}$. %, and hence, omitted to avoid repetition.
\end{proof}

%The result of Lemma~\ref{lem:NCJPLT2} is equivalent to the necessary (but not sufficient) condition for joint privacy in Section~\ref{sec:LINEAR} for linear JPLT-II protocols. 

%The result of Lemma~\ref{lem:NCJPLT2} is general and applies to any (linear or non-linear) JPLT-II protocol. Notwithstanding, this result establishes an interesting connection between linear JPLT-II protocols and linear codes: for any linear JPLT-II protocol, the coefficient matrix of the server's answer to the user's query must generate a (linear) code of length $K$ such that, when punctured at any $K-D$ coordinates, the resulting code contains a group of $L$ codewords that are \emph{linearly independent}. 
%Note that this condition is weaker than the one in Lemma~\ref{lem:NCJPLT1} for JPLT-I protocols, because it only requires having groups of linearly independent codewords, instead of MDS codewords.
%Recall that a similar, yet stronger, condition must hold for any linear JPLT-I protocol (see Lemma~\ref{lem:NCJPLT1}). % as each of the underlying groups of $L$ codewords must be MDS (instead of only linearly independent). 

\begin{lemma}\label{lem:JPLT2-Conv}
The rate of any JPLT-II protocol for $K$ messages, demand's support size $D$, and demand's dimension $L$ is upper bounded by $L/(K-D+L)$.
\end{lemma}

\begin{proof}
Consider an arbitrary JPLT-II protocol that generates a pair $(\mathrm{Q}^{[\mathrm{W},\mathrm{V}]},\mathrm{A}^{[\mathrm{W},\mathrm{V}]})$ for any given $(\mathrm{W},\mathrm{V})$. 
We denote $\mathbf{Q}^{[\mathbf{W},\mathbf{V}]}$ and $\mathbf{A}^{[\mathbf{W},\mathbf{V}]}$ by $\mathbf{Q}$ and $\mathbf{A}$, respectively. 
To show the rate upper bound, we need to show that ${H(\mathbf{A})\geq (K-D+L)B}$, where $B = N\log_2 q$ is the entropy of a uniformly distributed message over $\mathbbmss{F}_q^{N}$. 
Let $T\triangleq C_{K,D}$. Consider an arbitrary ordering of all elements in $\mathbbm{W}$, say, $\mathrm{W}_1,\dots,\mathrm{W}_{T}$, where $\mathrm{W}_i$'s are distinct $D$-subsets of $[K]$. 
By Lemma~\ref{lem:NCJPLT2}, there exist $n_i\geq 1$ matrices ${\mathrm{V}^{1}_i,\dots,\mathrm{V}^{n_i}_i\in \mathbbm{V}_{I\hspace{-0.04cm}I}}$ for ${i\in [T]}$, each of rank $L$, such that ${H(\mathbf{Z}^{j}_i |\mathbf{A},\mathbf{Q}) = 0}$ for ${j\in [n_i]}$, where ${\mathbf{Z}^{j}_i\triangleq \mathbf{Z}^{[\mathrm{W}_i,\mathrm{V}^{j}_i]}}$. 
Thus, $H(\mathbf{Z}^{1}_1,\dots,\mathbf{Z}^{n_1}_1,\dots,\mathbf{Z}^{1}_{T},\dots,\mathbf{Z}^{n_T}_T|\mathbf{A},\mathbf{Q})=0$.
%\begin{equation}\label{eq:6}
%H(\mathbf{Z}^{1}_1,\dots,\mathbf{Z}^{n_1}_1,\dots,\mathbf{Z}^{1}_{T},\dots,\mathbf{Z}^{n_T}_T|\mathbf{A},\mathbf{Q})=0.    
%\end{equation}
Similarly as in~\eqref{eq:1}-\eqref{eq:3}, %using~\eqref{eq:6} 
we can then show that
\begin{equation}
    H(\mathbf{A})\geq H(\mathbf{Z}^{1}_1,\dots,\mathbf{Z}^{n_1}_1,\dots,\mathbf{Z}^{1}_{T},\dots,\mathbf{Z}^{n_T}_T).\label{eq:7}
\end{equation} 
In the following, we lower bound the right hand-side of~\eqref{eq:7}.   %$H(\mathbf{Z}^{1}_1,\dots,\mathbf{Z}^{n_1}_1,\dots,\mathbf{Z}^{1}_{T},\dots,\mathbf{Z}^{n_T}_T)$, we proceed as follows. 

For any $i\in [T]$ and ${j\in [n_i]}$, let $\mathrm{U}^{j}_{i}$ be the global coefficient matrix of the potential demand $\mathbf{Z}^{j}_i$, i.e., 
the submatrix of $\mathrm{U}^{j}_{i}$ formed by the columns indexed by $\mathrm{W}_i$ is equal to $\mathrm{V}^{j}_{i}$, and the rest of the columns of $\mathrm{U}^{j}_{i}$ are all-zero. 
Note that the rank of $\mathrm{U}^{j}_{i}$ is $L$. 
This is simply because $\mathrm{U}^{j}_{i}$ has $L$ rows, and it contains the matrix $\mathrm{V}^j_i$ of rank $L$ as a submatrix.
Consider the $TL\times K$ matrix formed by vertically concatenating the $L\times K$ matrices $\mathrm{U}^{1}_{1},\dots,\mathrm{U}^{n_1}_{1},\dots,\mathrm{U}^{1}_{T},\dots,\mathrm{U}^{n_T}_{T}$. 
Choose an arbitrary basis of the row space of this $TL\times K$ matrix, and let $\mathrm{M}$ be a matrix formed by the chosen basis. 
It is easy to see that each row of $\mathbf{Z}^{j}_{i}$ can be written as a linear combination of the rows of the $TL\times N$ matrix $\mathrm{M}\mathbf{X}$, where $\mathbf{X}=[\mathbf{X}^{\transpose}_1,\dots,\mathbf{X}^{\transpose}_K]^{\transpose}$ is the $K\times N$ matrix of messages.
Since $\mathbf{X}_1,\dots,\mathbf{X}_K$ are independently and uniformly distributed over $\mathbbmss{F}^{N}_q$, we have
\begin{equation}
    H(\mathbf{Z}^{1}_1,\dots,\mathbf{Z}^{n_1}_1,\dots,\mathbf{Z}^{1}_{T},\dots,\mathbf{Z}^{n_T}_T) = \mathrm{rank}(\mathrm{M})\times B.\label{eq:8}
\end{equation} 
Combining~\eqref{eq:7} and~\eqref{eq:8}, we have $H(\mathbf{A})\geq \mathrm{rank}(\mathrm{M})\times B$. 
Recall that we need show that $H(\mathbf{A})\geq (K-D+L)B$. 
Thus, it suffices to show that $\mathrm{rank}(\mathrm{M})\geq K-D+L$. 
We prove this by the way of contradiction. 

Let $m\triangleq \mathrm{rank}(\mathrm{M})$. 
Note that $m\geq L$. 
This is because by the recoverability condition, the row space of $\mathrm{M}$ must contain the rows of the matrix $\mathrm{U}$, where $\mathrm{U}$ is the global coefficient matrix of the user's demand $\mathbf{Z} = \mathrm{U}\mathbf{X}$, and the rank of $\mathrm{U}$ is $L$. 
Suppose that ${m<K-D+L}$.
Choose an arbitrary basis of the row space of $\mathrm{M}$, and let $\tilde{\mathrm{M}}$ be an $m\times K$ matrix formed by this basis. 
By performing Gauss-Jordan elimination on a properly chosen column-permutation of $\tilde{\mathrm{M}}$, we can obtain a matrix of the form $[\mathrm{I},\mathrm{P}]$, where $\mathrm{I}$ is an $m\times m$ identity matrix, and $\mathrm{P}$ is an $m\times (K-m)$ matrix. 
By the construction of $[\mathrm{I},\mathrm{P}]$, the rows of $\mathrm{U}_i^j$ must be in the row space of $[\mathrm{I},\mathrm{P}]$. 
Without loss of generality, assume that $[\mathrm{I},\mathrm{P}]$ is obtained by performing elimination on $\tilde{\mathrm{M}}$ (instead of a column-permutation of $\tilde{\mathrm{M}}$). 

Next, we prove that $\mathrm{P}$ is an all-zero matrix. 
Let ${S\triangleq C_{K-m,D-L}}$, and let $\hat{\mathrm{W}}_1,\dots,\hat{\mathrm{W}}_S$ be all $(D-L)$-subsets of $[m+1:K]$. 
Without loss of generality, assume that ${\mathrm{W}_{i} = [L] \cup \hat{\mathrm{W}}_i}$ for $i\in [S]$. 
Since $m\geq L$, $[L]$ and $\hat{\mathrm{W}}_i$ are disjoint, and $|\mathrm{W}_{i}| = L+(D-L)=D$. 
For arbitrary ${i\in [S]}$ and ${j\in [n_i]}$, consider the matrix $\mathrm{U}^{j}_{i}$, and the submatrix of $[\mathrm{I},\mathrm{P}]$ formed by the first $L$ rows, denoted by $[\hat{\mathrm{I}},0,\hat{\mathrm{P}}]$, where $\hat{\mathrm{I}}$ is an ${L\times L}$ identity matrix, $0$ is an ${L\times (m-L)}$ all-zero matrix, and $\hat{\mathrm{P}}$ is an ${L\times (K-m)}$ matrix. 
Note that $\mathrm{W}_i$ does not contain any index in ${[L+1:m]}$, and any row of $[\mathrm{I},\mathrm{P}]$ with an index in ${[L+1:m]}$, i.e., any row of $[\mathrm{I},\mathrm{P}]$ that is not included in $[\hat{\mathrm{I}},0,\hat{\mathrm{P}}]$, has a nonzero entry at a distinct column with an index in ${[L+1:m]}$. 
Thus, the rows of $\mathrm{U}^{j}_{i}$ must be in the row space of $[\hat{\mathrm{I}},0,\hat{\mathrm{P}}]$. 
Recall that $\mathrm{U}^{j}_{i}$ and $[\hat{\mathrm{I}},0,\hat{\mathrm{P}}]$ have $L$ rows. 
Since $\mathrm{W}_i = [L]\cup \hat{\mathrm{W}}_i$ and $\hat{\mathrm{W}}_i$ is a $(D-L)$-subset of $[m+1:K]$, the rows of $\mathrm{U}^{j}_{i}$ lie in the row space of $[\hat{\mathrm{I}},0,\hat{\mathrm{P}}]$ iff the submatrix of $\hat{\mathrm{P}}$ formed by the columns indexed by $[m+1:K]\setminus  \hat{\mathrm{W}}_i$ is an all-zero matrix. 
Using the same argument for all ${i\in [S]}$, it follows that $\hat{\mathrm{P}}$ is an all-zero matrix.
Re-defining $\mathrm{W}_1,\dots,\mathrm{W}_S$ by replacing $[L]$ with different $L$-subsets of $[m]$, 
%instead of taking only $[L]$, 
and repeating the same arguments as above, it follows that $\mathrm{P}$ is an all-zero matrix. %, as was to be shown.  

Now, we can simply arrive at a contradiction. 
Recall that by assumption $m<K-D+L$, or equivalently, $K-m\geq D-L$, and $\mathrm{P}$ has $K-m$ columns.  
Without loss of generality, assume that $\mathrm{W}_1 = {[L-1]\cup [m+1:m+D-L+1]}$. 
Note that $|\mathrm{W}_1| = D$. 
Consider the matrix $\mathrm{U}_1^{1}$. 
Recall that $\mathrm{rank}(\mathrm{U}_1^{1})=L$, and the rows of $\mathrm{U}_1^{1}$ must lie in the row space of $[\mathrm{I},\mathrm{P}]$, or particularly, in the row space of the submatrix of $[\mathrm{I},\mathrm{P}]$ formed by the first $L-1$ rows (all rows of $\mathrm{U}_1^{1}$ are linearly independent of the last $m-L+1$ rows of $[\mathrm{I},\mathrm{P}]$).
The rank of this submatrix is $L-1$, and this is a contradiction because $\mathrm{rank}(\mathrm{U}_1^{1})=L$. 
Thus, ${m\geq K-D+L}$.
%This completes the proof.  
\end{proof}

\subsection{Achievability Scheme}\label{subsec:JPLT2-Ach}
In this section, we present a JPLT-II protocol, termed the \emph{Specialized Augmented Code protocol}, which is capacity-achieving for any $q\geq K$. 
An illustrative example of this protocol can be found in Appendix~\ref{app:2}.

The Specialized Augmented Code protocol consists of three steps as described below.
\vspace{0.125cm}

\textbf{Step 1:} Given $\mathrm{W}\in \mathbbm{W}$ and $\mathrm{V}=[\mathrm{v}_1^{\transpose},\dots,\mathrm{v}_l^{\transpose}]^{\transpose}\in \mathbbm{V}_{I\hspace{-0.04cm}I}$, the user constructs a $(K-D+L)\times K$ matrix $\mathrm{G}$, and sends $\mathrm{G}$ as the query $\mathrm{Q}^{[\mathrm{W},\mathrm{V}]}$ to the server. 
To construct $\mathrm{G}$, the user first constructs the global coefficient matrix $\mathrm{U}$ from $\mathrm{V}$.  
Next, the user constructs a ${(K-D+L)\times K}$ matrix $\hat{\mathrm{G}}$ by vertically concatenating the $L\times K$ matrix $\mathrm{U}$ and an arbitrary $(K-D)\times K$ MDS matrix $\mathrm{M}$---generated independently from $\mathrm{W}$ and $\mathrm{V}$.
(For any $q\geq K$, the matrix $\mathrm{M}$ can be constructed as the generator matrix of a $[K,K-D]$ GRS code over $\mathbbmss{F}_q$, with arbitrary nonzero multipliers and distinct evaluation points.)
That is, $\hat{\mathrm{G}} = [\mathrm{U}^{\transpose},\mathrm{M}^{\transpose}]^{\transpose}$. 
Observe that the code generated by $\hat{\mathrm{G}}$ is the result of augmenting the code generated by $\mathrm{U}$ with the codewords of the MDS code generated by $\mathrm{M}$. 
The user then constructs the matrix $\mathrm{G}$ by multiplying the matrix $\hat{\mathrm{G}}$ by a randomly generated $(K-D+L)\times (K-D+L)$ invertible matrix $\mathrm{R}$, i.e., ${\mathrm{G} = \mathrm{R}\hat{\mathrm{G}}}$.
Note that the matrix $\mathrm{G}$ does not necessarily generate an MDS code, and this protocol may not serve as a JPLT-I protocol in general.   

\textbf{Step 2:} Given the query $\mathrm{Q}^{[\mathrm{W},\mathrm{V}]}$, i.e., the matrix $\mathrm{G}$, the server computes the $(K-D+L)\times N$ matrix $\mathrm{Y}\triangleq \mathrm{G}\mathrm{X}$, and sends $\mathrm{Y}$ as the answer $\mathrm{A}^{[\mathrm{W},\mathrm{V}]}$ back to the user.\vspace{0.125cm}  

\textbf{Step 3:} Upon receiving the answer $\mathrm{A}^{[\mathrm{W},\mathrm{V}]}$, i.e., the matrix $\mathrm{Y}$, the user computes the $(K-D+L)\times N$ matrix $\tilde{\mathrm{Y}} \triangleq \mathrm{R}^{-1}\mathrm{Y}$, and recovers the $l$th row of the demand matrix $\mathrm{Z}^{[\mathrm{W},\mathrm{V}]}$, i.e., $\mathrm{v}_l\mathrm{X}_{\mathrm{W}}$, for ${l\in [L]}$, from the $l$th row of $\tilde{\mathrm{Y}}$.

\begin{lemma}\label{lem:JPLT2-Ach}
The Specialized Augmented Code protocol is a JPLT-II protocol, and achieves the rate $L/(K-D+L)$. 
\end{lemma}

\begin{proof}
Similar to Lemma~\ref{lem:JPLT1-Ach}, to prove that the rate of this protocol is $L/(K-D+L)$, it suffice to show that the $(K-D+L)\times K$ matrix $\mathrm{G}$ has full rank, i.e., $\mathrm{rank}(\mathrm{G})=K-D+L$. 
Since $\mathrm{G} = \mathrm{R}\hat{\mathrm{G}}$ and $\mathrm{R}$ is invertible, we need to show that  $\mathrm{rank}(\hat{\mathrm{G}})=K-D+L$. 
Recall that $\hat{\mathrm{G}} = [\mathrm{U}^{\transpose},\mathrm{M}^{\transpose}]^{\transpose}$.
Each row of $\mathrm{U}$ has at most $D$ nonzero entries. 
However, the row space of $\mathrm{M}$ does not contain any row-vector with less than $D+1$ nonzero entries. 
This is because $\mathrm{M}$ generates a $[K,K-D]$ MDS code with the minimum distance ${K-(K-D)+1=D+1}$. 
By these arguments, the rows of $\mathrm{U}$ do not lie in the row space of $\mathrm{M}$, and hence,  $\mathrm{rank}(\hat{\mathrm{G}})=\mathrm{rank}(\mathrm{M})+\mathrm{rank}(\mathrm{U})$.
Obviously, $\mathrm{rank}(\mathrm{M})=K-D$ because $\mathrm{M}$ is a $(K-D)\times K$ MDS matrix, and $\mathrm{rank}(\mathrm{U})=L$ because $\mathrm{U}$ contains $\mathrm{V}$ as a submatrix, and $\mathrm{rank}(\mathrm{V})=L$ (by assumption).
Thus, $\mathrm{rank}(\hat{\mathrm{G}})=K-D+L$, as was to be shown. 

Next, we prove that the joint privacy condition is satisfied. 
To this end, we show that, for any $\tilde{\mathrm{W}}\in \mathbbm{W}$, the row space of $\mathrm{G}$, or equivalently, the row space of $\hat{\mathrm{G}}$, contains a unique $L$-dimensional subspace on the coordinates indexed by $\tilde{\mathrm{W}}$. 
Without loss of generality, assume that $\tilde{\mathrm{W}} ={[K-D+1:K]}$. 
We can rewrite the matrix $\hat{\mathrm{G}}$ as 
\begin{equation*}
\hat{\mathrm{G}} = 
\begin{bmatrix} 
\mathrm{U}_1 & \mathrm{U}_2\\ 
\mathrm{M}_1 & \mathrm{M}_2 
\end{bmatrix},    
\end{equation*} where $\mathrm{U}_1$ (or $\mathrm{M}_1$) is an ${L\times (K-D)}$ submatrix of $\mathrm{U}$ (or $\mathrm{M}$) formed by the columns indexed by $[K]\setminus \tilde{\mathrm{W}}$, and 
$\mathrm{U}_2$ (or $\mathrm{M}_2$) is an $L\times D$ submatrix of $\mathrm{U}$ (or $\mathrm{M}$) formed by the columns indexed by $\tilde{\mathrm{W}}$.
Since $\mathrm{M}_1$ is a $(K-D)\times (K-D)$ submatrix of $\mathrm{M}$, and $\mathrm{M}$ is a $(K-D)\times K$ MDS matrix, the row space of $\mathrm{M}_1$ is a $(K-D)$-dimensional subspace on the $K-D$ coordinates indexed by $[K]\setminus \tilde{\mathrm{W}}$. 
This implies that each row of $\mathrm{U}_1$ can be written as a unique linear combination of the rows of $\mathrm{M}_1$. 
Thus, by performing Gauss-Jordan elimination on the matrix $\hat{\mathrm{G}}$, we can obtain a $(K-D+L)\times K$ matrix $\tilde{\mathrm{G}}$ given by
\begin{equation*}
\tilde{\mathrm{G}}=
\begin{bmatrix} 
0 & \tilde{\mathrm{U}}\\ 
\mathrm{I} & \tilde{\mathrm{M}} 
\end{bmatrix},    
\end{equation*} where $0$ is an ${L\times (K-D)}$ all-zero matrix, $\mathrm{I}$ is a ${(K-D)\times (K-D)}$ identity matrix, $\tilde{\mathrm{U}}$ is an ${L\times D}$ matrix, and $\tilde{\mathrm{M}}$ is a ${(K-D)\times D}$ matrix.
Note that $\mathrm{rank}(\tilde{\mathrm{G}})=\mathrm{rank}(\hat{\mathrm{G}})=K-D+L$, and $\mathrm{rank}(\tilde{\mathrm{G}})=\mathrm{rank}(\mathrm{I})+\mathrm{rank}(\tilde{\mathrm{U}})= (K-D)+\mathrm{rank}(\tilde{\mathrm{U}})$. 
Thus, ${\mathrm{rank}(\tilde{\mathrm{U}})=L}$.
This implies that the row space of the $L\times K$ matrix $[0,\tilde{\mathrm{U}}]$ is an $L$-dimensional subspace on the coordinates indexed by $\tilde{\mathrm{W}}=[K-D+1:K]$. 
Moreover, this subspace is unique because the $(K-D)\times K$ matrix $[\mathrm{I},\tilde{\mathrm{M}}]$ is MDS. 
Note, also, that from the perspective of the server, each of these $L$-dimensional subspaces (corresponding to a distinct $\tilde{\mathrm{W}}\in \mathbbm{W}$) is equally likely to be the subspace spanned by the rows of the demand's global coefficient matrix $\mathrm{U}$. 
Thus, given the query (i.e., the matrix $\mathrm{G}$), every $\tilde{\mathrm{W}}\in \mathbbm{W}$ is equally likely to be the demand's support. 
This completes the proof of joint privacy. 

The proof of recoverability is straightforward. 
By Step~3 of the protocol, $\tilde{\mathrm{Y}}=\mathrm{R}^{-1}\mathrm{Y}$. 
Rewriting $\mathrm{Y}$ as $\mathrm{G}\mathrm{X}=\mathrm{R}\hat{\mathrm{G}}\mathrm{X}$, it follows that $\tilde{\mathrm{Y}}=\hat{\mathrm{G}}\mathrm{X}=[(\mathrm{U}\mathrm{X})^{\transpose},(\mathrm{M}\mathrm{X})^{\transpose}]^{\transpose}$.
This shows that the $L\times N$ submatrix of $\tilde{\mathrm{Y}}$ formed by the first $L$ rows is equal to the demand matrix $\mathrm{U}\mathrm{X}=\mathrm{V}\mathrm{X}_{\mathrm{W}}$.
\end{proof}

\section{Conclusion and Future Work}\label{sec:Con}
In this work, we introduced the problem of Private Linear Transformation (PLT) which generalizes the Private Information Retrieval (PIR) and Private Linear Computation (PLC) problems. 
The PLT problem includes a dataset that is stored on a single (or multiple) remote server(s), and a user who wishes to compute multiple linear combinations of a subset of items belonging to the dataset.
The goal is to perform the computation such that the total amount of information downloaded is minimized, while the identities of items required for the computation are kept private.

We focused on the single-server setting of the PLT problem with joint privacy guarantees, referred to as the JPLT problem. 
The notion of joint privacy ensures that the identities of all items required for the computation are protected jointly. % referred to as the JPLT problem. 
%In this problem, referred to as the JPLT problem, the privacy requirement is to protect the identities of all items required for the computation jointly. 
We considered two different models, depending on whether the coefficient matrix of the required linear combinations is MDS. 
For each model, we characterized the capacity, where the capacity is defined as the supremum of all achievable download rates. 
In addition, we presented a capacity-achieving scheme for each of the models being considered.  

There remain several open problems---closely related to the JPLT problem. 
Below, we list a few of these problems.  
\begin{enumerate}
    %\item The capacity of JPLT for the settings in which the coefficient matrix of the required linear combinations has full rank and no all-zero column remains unknown. 
    %Such matrices are particularly of interest in the scenarios where 
    \item It was recently shown that, as compared to the single-server setting, PIR and PLC can be performed much more efficiently (in terms of the download rate) when there are multiple servers that store identical copies or coded versions of the dataset, see, e.g.,~\cite{SJ2017,BU2018,SJ2018,SJ2018No2,TGKHHER2017,TER2017,TER2017,BU18,OK2018,OLRK2018}. 
    %Similar results were also shown for the cases in which the servers can collude with some physical limitations (see, e.g.,~\cite{SJ2018No2,TGKHHER2017,TER2017}), and the cases in which each server %(without any collusion or with limited capability for collusion)
    %that 
    %stores %identical copies or 
    %a coded version of the messages in the dataset (see, e.g.,~\cite{TER2017,BU18,OK2018,OLRK2018}). 
    %These results motivate the study of PLT in the multi-server settings.
    Motivated by these results, an important direction for research is to characterize the capacity of multi-server PLT with joint privacy guarantees.
    \item Establishing the fundamental limits of (single-server or multi-server) PLT with joint privacy in the presence of a prior side information is another direction for future work. 
    This is motivated by the recent developments in PIR and PLC with side information, see, e.g.,~\cite{KGHERS2020,HKS2019Journal,HKRS2019,KKHS32019,HKGRS2018,LG2018,HKS2018,HKS2019,T2017,WBU2018,WBU2018No2,CWJ2020,SSM2018,KKHS22019,KKHS12019}.
    %\item Several new types of information-theoretic privacy were recently studied for the PIR problem in~\cite{LKRAY2020,SATL2020_2,TKKN2021}. 
    %These works motivate the study of the PLC problem (and more generally, the PLT problem) under these types of 
    \item Many machine learning and cloud computing algorithms require computing non-linear functions on a subset of dataset. 
    For instance, evaluating polynomials on a subset of training samples finds application in distributed stochastic gradient descent for linear regression~\cite{LCC}. 
    The need for protecting the data access privacy in such scenarios motivates the problem of designing efficient privacy-preserving schemes for non-linear function computation. 
\end{enumerate}

\appendix %[Illustrative Examples of the Proposed Protocols]

\subsection{An Example of the Specialized MDS Code Protocol}\label{app:1}

%\begin{example}
%\normalfont 
Consider  a  scenario in which the server has ${K=10}$ messages $\mathrm{X}_1,\dots,\mathrm{X}_{10}\in\mathbbmss{F}^{N}_{11}$ for an arbitrary integer $N\geq 1$, and the user wishes to compute ${L=2}$ linear combinations of ${D=5}$ messages $X_2$, $X_4$, $X_5$, $X_7$, $X_8$, say, 
%\begin{align*}
%Z_1 &= X_2+3X_4+2X_5+X_7+6X_8,\\
%Z_2 &= 3X_2+10X_4+7X_5+4X_7+8X_8.
%\end{align*}
$Z_1 = X_2+3X_4+2X_5+X_7+6X_8$, and
$Z_2 = 3X_2+10X_4+7X_5+4X_7+8X_8$. 
For this example, $\mathrm{W}=\{2,4,5,7,8\}$, and %$\mathrm{V}$ is given by
\begin{equation*}
\mathrm{\mathrm{V}} = 
\begin{bmatrix}
1 & 3 & 2 & 1 & 6\\
3 & 10 & 7 & 4 & 8\\
\end{bmatrix}.
\end{equation*}
%$\mathrm{V}=[\mathrm{v}_{1},\mathrm{v}_{2}]^{\transpose}$ 
%where $\mathrm{v}_{2}=[3,10,7,4,8]$.  
It is easy to verify that $\mathrm{V}$ generates a $[5,2]$ GRS code with the multipliers $\{v_{1},\dots,v_{5}\}=\{1,3,2,1,6\}$ and the evaluation points $\{\omega_{1},\dots,\omega_{5}\}=\{3,7,9,4,5\}$. Then, the user obtains the parity-check matrix $\myLambda$ of the code generated by $\mathrm{V}$ as 
\begin{equation*}
\mathrm{\myLambda} = 
\begin{bmatrix}
3 & 10 & 8 & 8 & 7\\
9 & 4 & 6 & 10 & 2\\
5 & 6 & 10 & 7 & 10 \\
\end{bmatrix}.
\end{equation*}
Note that $\myLambda$ generates a $[5,3]$ MDS code with the multipliers $\{\lambda_1,\dots,\lambda_5\} = \{3,10,8,8,7\}$ and the evaluation points $\{\omega_1,\dots,\omega_5\} = \{3,7,9,4,5\}$. 

Next, the user extends the ${3\times 5}$ matrix $\myLambda$ to a ${3\times 10}$ matrix $\mathrm{H}$ that satisfies the conditions (i) and (ii) specified in Step 1 of the Specialized GRS Code protocol. 
%Let $\mathrm{v}_{1}=[1,3,2,1,6]$
Suppose the user randomly chooses $6$ additional multipliers $\{\lambda_{6},\dots,\lambda_{10}\}=\{3,5,1,1,4\}$ from ${\mathbbmss{F}_{11}\setminus \{0\}}$, and $6$ additional evaluation points $\{\omega_{6},\dots,\omega_{10}\}=\{6,1,10,2,8\}$ from ${\mathbbmss{F}_{11}\setminus \{\omega_1,\dots,\omega_{5}\}}$. Followed by constructing a permutation $\pi$ as described in Step 1 of the Specialized GRS Code protocol, say, $\{\pi(1),\dots,\pi(10)\}=\{2,4,5,7,8,1,3,6,9,10\}$, the user constructs the extended matrix $\mathrm{H}$ as 
%In this example, $\mathrm{W}=\{i_{1},\dots,i_{5}\}=\{2,4,5,7,8\}$ and hence $[K]\setminus \mathrm{W}=\{1,3,6,9,10\}$ where $[K]=\{1,\dots,10\}$. Let $\pi(1),\dots,\pi(10)=i_{1},\dots,i_{10}=2,4,5,7,8,1,3,6,9,10$. Then, the parity check matrix $H$ is
\begin{equation*}
\mathrm{H} = 
\begin{bmatrix}
3 & \mathbf{3} & 5 & \mathbf{10} & \mathbf{8} & 1 & \mathbf{8} & \mathbf{7} & 1 & 4\\
7 & \mathbf{9} & 5 & \mathbf{4} & \mathbf{6} & 10 & \mathbf{10} & \mathbf{2} & 2 & 10\\
9 & \mathbf{5} & 5 & \mathbf{6} & \mathbf{10} & 1 & \mathbf{7} & \mathbf{10} & 4 & 3
\end{bmatrix},
\end{equation*} where the columns of $\mathrm{H}$ indexed by $\pi(1)$, $\pi(2)$, $\pi(3)$, $\pi(4)$, $\pi(5)$ (i.e., the columns $2,4,5,7,8$) correspond to the columns $1,2,3,4,5$ of $\myLambda$, respectively, and the columns of $\mathrm{H}$ indexed by $\pi(6)$, $\pi(7)$, $\pi(8)$, $\pi(9)$, $\pi(10)$ (i.e., the columns $1,3,6,9,10$) correspond to the columns of the generator matrix of a $[5,3]$ GRS code with the multipliers $\{\lambda_6,\dots,\lambda_{10}\}$ and the evaluation points  $\{\omega_6,\dots,\omega_{10}\}$. 
That is, the $\pi(i)$th column of $\mathrm{H}$ for $i\in \{6,\dots,10\}$ is given by $[\lambda_{i},\lambda_i\omega_i,\lambda_i\omega_i^2]^{\transpose}$. %For instance, the column $\pi(6)$ (i.e., the column $1$) of $\mathrm{H}$ is given by $[\beta_{6},\beta_6\omega_6,\beta_6\omega_6^2]^{\transpose} = [3,7,9]^{\transpose}$. 
Since $\mathrm{H}$ generates a $[10,3]$ GRS code with the multipliers $\{\lambda_6,\lambda_1,\lambda_7,\lambda_2,\lambda_3,\lambda_8,\lambda_4,\lambda_5,\lambda_9,\lambda_{10}\}$ and the evaluation points $\{\omega_6,\omega_1,\omega_7,\omega_2,\omega_3,\omega_8,\omega_4,\omega_5,\omega_9,\omega_{10}\}$, $\mathrm{H}$ can be thought of as the parity-check matrix of a $[10,7]$ GRS code with the multipliers $\alpha_6=9$, $\alpha_1=10$, $\alpha_7=2$, $\alpha_2=7$, $\alpha_3=3$, $\alpha_8=1$, $\alpha_4=5$, $\alpha_5=4$, $\alpha_9=9$, $\alpha_{10}=9$ and the evaluation points $\omega_6=6$, $\omega_1=3$, $\omega_7=1$, $\omega_2=7$, $\omega_3=9$, $\omega_8=10$, $\omega_4=4$, $\omega_5=5$, $\omega_9=2$, $\omega_{10}=8$. (The process of computing $\alpha_i$'s is explained in Step 1 of the Specialized GRS Code protocol.) The user then obtains the generator matrix $\mathrm{G}$ of this code, 
%In this example, $\{\alpha_{1},\dots,\alpha_{10}\}=\{10,7,4,5,4,9,2,1,4,4\}$ and $\{\omega_{1},\dots,\omega_{10}\}=\{6,3,1,7,9,10,4,5,2,8\}$
\begin{equation*}
\mathrm{G} = 
\begin{bmatrix}
9 & 10 & 2 & 7 & 3 & 1 & 5 & 4 & 9 & 9\\
10 & 8 & 2 & 5 & 5 & 10 & 9 & 9 & 7 & 6\\
5 & 2 & 2 & 2 & 1 & 1 & 3 & 1 & 3 & 4 \\
8 & 6 & 2 & 3 & 9 & 10 & 1 & 5 & 6 & 10\\
4 & 7 & 2 & 10 & 4 & 1 & 4 & 3 & 1 & 3\\
2 & 10 & 2 & 4 & 3 & 10 & 5 & 4 & 2 & 2\\
1 & 8 & 2 & 6 & 5 & 1 & 9 & 9 & 4 & 5\\
\end{bmatrix}.
\end{equation*} Then, the user sends the matrix $\mathrm{G}$ as the query to the server. The server then computes the matrix ${\mathrm{Y}= \mathrm{G}\mathrm{X}}$, and sends it back to the user. 
\begin{comment}
\begin{equation*}
\mathrm{y} = 
\begin{bmatrix}
9\mathrm{X}_{1} + 10\mathrm{X}_{2} + 2\mathrm{X}_{3} + 7\mathrm{X}_{4} + 4\mathrm{X}_{5} + \mathrm{X}_{6} + 5\mathrm{X}_{7} + 4\mathrm{X}_{8} + 4\mathrm{X}_{9} + 4\mathrm{X}_{10}\\
10\mathrm{X}_{1} + 8\mathrm{X}_{2} + 2\mathrm{X}_{3} + 5\mathrm{X}_{4} + 3\mathrm{X}_{5} + 10\mathrm{X}_{6} + 9\mathrm{X}_{7} + 9\mathrm{X}_{8} + 8\mathrm{X}_{9} + 10\mathrm{X}_{10}\\
5\mathrm{X}_{1} + 2\mathrm{X}_{2} + 2\mathrm{X}_{3} + 2\mathrm{X}_{4} + 5\mathrm{X}_{5} + \mathrm{X}_{6} + 3\mathrm{X}_{7} + \mathrm{X}_{8} + 5\mathrm{X}_{9} + 3\mathrm{X}_{10} \\
8\mathrm{X}_{1} + 6\mathrm{X}_{2} + 2\mathrm{X}_{3} + 3\mathrm{X}_{4} + 1\mathrm{X}_{5} + 10\mathrm{X}_{6} + \mathrm{X}_{7} + 5\mathrm{X}_{8} + 10\mathrm{X}_{9} + 2\mathrm{X}_{10}\\
4\mathrm{X}_{1} + 7\mathrm{X}_{2} + 2\mathrm{X}_{3} + 10\mathrm{X}_{4} + 9\mathrm{X}_{5} + \mathrm{X}_{6} + 4\mathrm{X}_{7} + 3\mathrm{X}_{8} + 9\mathrm{X}_{9} + 5\mathrm{X}_{10}\\
2\mathrm{X}_{1} + 10\mathrm{X}_{2} + 2\mathrm{X}_{3} + 4\mathrm{X}_{4} + 4\mathrm{X}_{5} + 10\mathrm{X}_{6} + 5\mathrm{X}_{7} + 4\mathrm{X}_{8} + 7\mathrm{X}_{9} + 7\mathrm{X}_{10}\\
1\mathrm{X}_{1} + 8\mathrm{X}_{2} + 2\mathrm{X}_{3} + 6\mathrm{X}_{4} + 3\mathrm{X}_{5} + \mathrm{X}_{6} + 9\mathrm{X}_{7} + 9\mathrm{X}_{8} + 3\mathrm{X}_{9} + \mathrm{X}_{10}\\
\end{bmatrix}.
\end{equation*}
\end{comment}
Next, the user constructs two polynomials 
\begin{align*}
f_1(x)& ={(x-\omega_6)(x-\omega_7)(x-\omega_8)(x-\omega_9)(x-\omega_{10})}\\
& = {(x-6)(x-1)(x-10)(x-2)(x-8)},    
\end{align*}
and $f_2(x)=x f_1(x)$. The coefficient vectors of the polynomials $f_1(x)$ and $f_2(x)$ are given by $\mathrm{c}_1=[8,1,8,9,6,1,0]^{\transpose}$ and $\mathrm{c}_2=[0,8,1,8,9,6,1]^{\transpose}$, respectively.
% $\mathrm{c}_1$ and $\mathrm{c}_2$
%$f_1(x)=x^{5}+6x^{4}+9x^{3}+8x^{2}+x^{1}+8$ 
%f_2(x)=x^{6}+6x^{5}+9x^{4}+8x^{3}+1x^{2}+8x^{1}$.
The user then recovers their demand, i.e., $Z_1$ and $Z_2$, by computing %$\mathrm{c}_{1}^{\transpose}\mathrm{y}$ and $\mathrm{c}_{2}^{\transpose}\mathrm{y}$. It is easy to verify that  
\begin{align*}
& Z_1 = \mathrm{c}_{1}^{\transpose}\mathrm{Y} = X_{2}+3X_{4}+2X_{5}+X_{7}+6X_{8},\\
& Z_2 = \mathrm{c}_{2}^{\transpose}\mathrm{Y}={3X_{2}+10X_{4}+7X_{5}+4X_{7}+8X_{8}}.
\end{align*} For this example, the rate of the Specialized MDS Code protocol is ${L}/{(K-D+L)} = 2/7$, whereas the rate of a PIR-based scheme or a PLC-based scheme is ${L}/{K} = {2}/{10}$ or ${1}/{(K-D)} = {1}/{5}$, respectively.    
%$Z_1 = \mathrm{c}_{1}^{\transpose}\mathrm{y}={X_{2}+3X_{4}+2X_{5}+X_{7}+6X_{8}}$, and $Z_2 = \mathrm{c}_{2}^{\transpose}\mathrm{y}={3X_{2}+10X_{4}+7X_{5}+4X_{7}+8X_{8}}$.  %As we can see, these two linear combinations are exactly $\mathrm{v}^{\transpose}_1\mathrm{X}_{\mathrm{W}}$ and $\mathrm{v}^{\transpose}_2\mathrm{X}_{\mathrm{W}}$ as was desired. Hence, the user successfully recovers the two coded combinations (of the demand).
%\end{example}

\subsection{An Example of the Specialized Augmented Code Protocol}\label{app:2}

%\begin{example}
%\normalfont 
Consider  a  scenario in which the server has ${K=10}$ messages $X_1,\dots,X_{10}\in\mathbbmss{F}^{N}_{11}$ for an arbitrary integer $N\geq 1$, and the user wants to compute ${L=2}$ linear combinations of ${D=5}$ messages $X_2$, $X_4$, $X_5$, $X_7$, $X_8$, say, $Z_1 = 3X_2+X_4+6X_5+2X_7+6X_8$ and
$Z_2 = 10X_2+4X_4+8X_5+7X_7+9X_8$. For this example, $\mathrm{W}=\{2,4,5,7,8\}$, and 
\begin{equation*}
\mathrm{V} = 
\begin{bmatrix}
3 & 1 & 6 & 2 & 6\\
10 & 4 & 8 & 7 & 9\\
\end{bmatrix}.
\end{equation*}
Note that $\mathrm{V}$ has full rank, but it is not MDS.
First, the user constructs the demand's global coefficient matrix $\mathrm{U}$ as
\begin{equation*}
\mathrm{U} = 
\begin{bmatrix}
0 & 3 & 0 & 1 & 6 & 0 & 2 & 6 & 0 & 0\\
0 & 10 & 0 & 4 & 8 & 0 & 7 & 9 & 0 & 0\\
\end{bmatrix}.
\end{equation*}
%The columns of $\mathrm{U}$ indexed by $\mathrm{W}$ are exactly the columns of the matrix $\mathrm{V}$, and the rest of the columns of $\mathrm{U}$ are all-zero. 
Next, the user generates an arbitrary $5\times10$ MDS matrix $\mathrm{M}$, independently from $\mathrm{W}$ and $\mathrm{V}$. 
For this example, suppose the matrix $\mathrm{M}$ is given by
\begin{equation*}
\mathrm{M} = 
\begin{bmatrix}
2 & 1 & 4 & 7 & 9 & 1 & 10 & 5 & 4 & 3\\
6 & 5 & 3 & 5 & 3 & 6 & 10 & 6 & 10 & 6\\
7 & 3 & 5 & 2 & 1 & 3 & 10 & 5 & 3 & 1\\
10 & 4 & 1 & 3 & 4 & 7 & 10 & 6 & 2 & 2\\
8 & 9 & 9 & 10 & 5 & 9 & 10 & 5 & 5 & 4\\
\end{bmatrix}.
\end{equation*}
The user then constructs a $7\times 10$ matrix $\hat{\mathrm{G}}$ by vertically concatenating the matrices $\mathrm{U}$ and $\mathrm{M}$, i.e., $\hat{\mathrm{G}} = [\mathrm{U}^{\transpose},\mathrm{M}^{\transpose}]^{\transpose}$, 
\begin{equation*}
\hat{\mathrm{G}} = 
\begin{bmatrix}
0 & 3 & 0 & 1 & 6 & 0 & 2 & 6 & 0 & 0\\
0 & 10 & 0 & 4 & 8 & 0 & 7 & 9 & 0 & 0\\
2 & 1 & 4 & 7 & 9 & 1 & 10 & 5 & 4 & 3\\
6 & 5 & 3 & 5 & 3 & 6 & 10 & 6 & 10 & 6\\
7 & 3 & 5 & 2 & 1 & 3 & 10 & 5 & 3 & 1\\
10 & 4 & 1 & 3 & 4 & 7 & 10 & 6 & 2 & 2\\
8 & 9 & 9 & 10 & 5 & 9 & 10 & 5 & 5 & 4\\
\end{bmatrix}.
\end{equation*}
Then, the user randomly generates a $7\times7$ invertible matrix $\mathrm{R}$, and  
%For this example, suppose that $\mathrm{R}$ is given by 
%\begin{equation*}
%\mathrm{{R}} = 
%\begin{bmatrix}
%3 & 9 & 7 & 2 & 1 & 0 & 5\\
%2 & 6 & 2 & 6 & 1 & 3 & 6\\
%9 & 4 & 3 & 0 & 4 & 4 & 4\\
%4 & 8 & 7 & 2 & 7 & 9 & 1\\
%1 & 3 & 7 & 10 & 7 & 8 & 5\\
%9 & 5 & 7 & 8 & 4 & 2 & 5\\
%8 & 5 & 6 & 9 & 4 & 8 & 8\\
%\end{bmatrix}.
%\end{equation*}
%where 
\begin{comment}
\begin{equation*}
\mathrm{R^{-1}} = 
\begin{bmatrix}
4 & 7 & 0 & 0 & 4 & 5 & 10\\
2 & 10 & 3 & 2 & 10 & 2 & 4\\
3 & 5 & 7 & 2 & 0 & 10 & 5\\
2 & 1 & 1 & 9 & 6 & 0 & 5\\
4 & 2 & 6 & 5 & 7 & 8 & 5\\
7 & 4 & 5 & 8 & 7 & 5 & 5\\
6 & 3 & 3 & 0 & 0 & 1 & 1\\
\end{bmatrix},
\end{equation*}
\end{comment}
%Next, the user 
constructs a $7\times 10$ matrix ${\mathrm{G}} = \mathrm{R}\hat{\mathrm{G}}$. 
For this example, suppose that $\mathrm{G}$ is given by
\begin{equation*}
\mathrm{{G}} = 
\begin{bmatrix}
7 & 10 & 7 & 7 & 9 & 1 & 10 & 0 & 10 & 10\\
4 & 2 & 0 & 9 & 7 & 6 & 6 & 0 & 8 & 7\\
7 & 2 & 6 & 7 & 10 & 2 & 9 & 4 & 8 & 4\\
8 & 10 & 10 & 3 & 7 & 2 & 5 & 6 & 4 & 7\\
1 & 1 & 3 & 2 & 0 & 2 & 8 & 5 & 3 & 3\\
7 & 2 & 9 & 6 & 9 & 5 & 5 & 8 & 6 & 9\\
7 & 10 & 8 & 7 & 3 & 2 & 5 & 10 & 6 & 3\\
\end{bmatrix}.
\end{equation*}
Next, the user sends ${\mathrm{G}}$ to the server. 
%Note that the matrix $\mathrm{G}$, unlike the case of the Specialized MDS Code protocol, does not necessarily generate an MDS code. 
Given $\mathrm{G}$, the server computes the matrix $\mathrm{Y} = \mathrm{G}\mathrm{X}$, where $\mathrm{X}=[X^{\transpose}_1,\dots,X^{\transpose}_{10}]^{\transpose}$, 
\begin{comment}
\begin{equation*}
\mathrm{{Y}} = 
\begin{bmatrix}
6 & 10 & 7 & 7 & 9 & 1 & 10 & 0 & 10 & 10\\
4 & 2 & 0 & 9 & 7 & 6 & 6 & 0 & 8 & 7\\
10 & 2 & 6 & 7 & 10 & 2 & 9 & 4 & 8 & 4\\
7 & 10 & 10 & 3 & 7 & 2 & 5 & 6 & 4 & 7\\
1 & 1 & 3 & 2 & 0 & 2 & 8 & 5 & 3 & 3\\
9 & 2 & 9 & 6 & 9 & 5 & 5 & 8 & 6 & 9\\
2 & 10 & 8 & 7 & 3 & 2 & 5 & 10 & 6 & 3\\
\end{bmatrix}
 \begin{bmatrix}
X_{1}\\
X_{2}\\
X_{3}\\
X_{4}\\
X_{5}\\
X_{6}\\
X_{7}\\
X_{8}\\
X_{9}\\
X_{10}\\
\end{bmatrix}
\end{equation*}
\end{comment}
and sends $\mathrm{Y}$ back to the user. 
Given the matrix $\mathrm{Y}$, the user recovers their demand matrix  $[Z^{\transpose}_1,Z^{\transpose}_2]^{\transpose} = \mathrm{V}\mathrm{X}_{\mathrm{W}} = \mathrm{U}\mathrm{X}$ from the matrix formed by the first $2$ rows of the matrix $\tilde{\mathrm{Y}}=\mathrm{R}^{-1}\mathrm{Y} = (\mathrm{R}^{-1}\mathrm{G})\mathrm{X} = \hat{\mathrm{G}}\mathrm{X} = [(\mathrm{U}\mathrm{X})^{\transpose},(\mathrm{M}\mathrm{X})^{\transpose}]^{\transpose}$. 

For this example, the rate of the Specialized Augmented Code protocol is ${L}/{(K-D+L)} = {2}/{7}$, whereas the rate of a PIR-based scheme or a PLC-based scheme is ${L}/{K} = {2}/{10}$ or ${1}/{(K-D)} = {1}/{5}$, respectively.
%\end{example}

\bibliographystyle{IEEEtran}
\bibliography{PIR_PC_Refs}

\end{document}